\documentclass[a4paper, 12pt]{article}

\usepackage{amssymb}
\usepackage{amsmath}
\usepackage{amsthm}
\usepackage{verbatim}
\usepackage{hyperref}

\newcommand{\be}{\begin{equation}}
\newcommand{\ee}{\end{equation}}

\newcommand{\bC}{\mathbb{C}}
\newcommand{\bN}{\mathbb{N}}
\newcommand{\bR}{\mathbb{R}}

\newcommand{\cA}{\mathcal{A}}
\newcommand{\cB}{\mathcal{B}}
\newcommand{\cC}{\mathcal{C}}
\newcommand{\cK}{\mathcal{K}}
\newcommand{\cL}{\mathcal{L}}
\newcommand{\cM}{\mathcal{M}}
\newcommand{\cN}{\mathcal{N}}
\newcommand{\cR}{\mathcal{R}}
\newcommand{\cS}{\mathcal{S}}
\newcommand{\cU}{\mathcal{U}}
\newcommand{\cW}{\mathcal{W}}

\newcommand{\mfa}{\mathfrak{a}}
\newcommand{\mfc}{\mathfrak{c}}
\newcommand{\mfg}{\mathfrak{g}}
\newcommand{\mfgl}{\mathfrak{gl}}
\newcommand{\mfk}{\mathfrak{k}}
\newcommand{\mfm}{\mathfrak{m}}
\newcommand{\mfp}{\mathfrak{p}}
\newcommand{\mfu}{\mathfrak{u}}
\newcommand{\mfX}{\mathfrak{X}}

\newcommand{\bsone}{\boldsymbol{1}}
\newcommand{\bsLambda}{\boldsymbol{\Lambda}}
\newcommand{\bsTheta}{\boldsymbol{\Theta}}
\newcommand{\bsX}{\boldsymbol{X}}

\newcommand{\ri}{\mathrm{i}}

\newcommand{\dd}{\mathrm{d}}
\newcommand{\diag}{\mathrm{diag}}
\newcommand{\reg}{\mathrm{reg}}
\newcommand{\tr}{\mathrm{tr}}
\newcommand{\Id}{\mathrm{Id}}
\newcommand{\ad}{\mathrm{ad}}
\newcommand{\wad}{\widetilde{\mathrm{ad}}}
\newcommand{\Real}{\mathrm{Re}}
\newcommand{\Imag}{\mathrm{Im}}
\newcommand{\End}{\mathrm{End}}

\newcommand{\htheta}{{\hat{\theta}}}
\newcommand{\hTheta}{{\hat{\Theta}}}
\newcommand{\hbsTheta}{{\hat{\bsTheta}}}

\newcommand{\eps}{\varepsilon}

\newcommand{\half}{\frac{1}{2}}

\newcommand{\PD}[2]{\frac{\partial #1}{\partial #2}}

\numberwithin{equation}{section}

\theoremstyle{plain}
\newtheorem{THEOREM}{Theorem}
\newtheorem{LEMMA}[THEOREM]{Lemma}
\newtheorem{PROPOSITION}[THEOREM]{Proposition}

\begin{document}

\begin{center}
    \Large{\textbf{Lax representation of the hyperbolic van Diejen dynamics
    with two coupling parameters}}
\end{center}

\medskip
\begin{center}
    B.G.~PUSZTAI${}^{a, b}$ and T.F.~G\"ORBE${}^{c}$ \\
    \bigskip
    ${}^a$ Bolyai Institute, University of Szeged, \\
    Aradi v\'ertan\'uk tere 1, H-6720 Szeged, Hungary \\
    ${}^b$ MTA Lend\"ulet Holographic QFT Group, Wigner RCP, \\
    H-1525 Budapest 114, P.O.B. 49, Hungary \\
    e-mail: \texttt{gpusztai@math.u-szeged.hu} \\
    \medskip
    ${}^c$ Department of Theoretical Physics, University of Szeged, \\
    Tisza Lajos krt 84-86, H-6720 Szeged, Hungary \\
    e-mail: \texttt{tfgorbe@physx.u-szeged.hu}
\end{center}

\medskip
\begin{abstract}
In this paper, we construct a Lax pair for the classical hyperbolic 
van Diejen system with two independent coupling parameters. Built upon 
this construction, we show that the dynamics can be solved by a projection 
method, which in turn allows us to initiate the study of the scattering 
properties. As a consequence, we prove the equivalence between the first 
integrals provided by the eigenvalues of the Lax matrix and the family of 
van Diejen's commuting Hamiltonians. Also, at the end of the paper, we 
propose a candidate for the Lax matrix of the hyperbolic van Diejen 
system with three independent coupling constants.



\end{abstract}

\renewcommand{\contentsname}{\large{Contents}}
\tableofcontents

\section{Introduction}
\label{SECTION_Introduction}
The Ruijsenaars--Schneider--van Diejen (RSvD) systems, or simply van Diejen 
systems \cite{van_Diejen_ComposMath, van_Diejen_TMP1994, van_Diejen_JMP1995}, 
are multi-parametric integrable deformations of the translation invariant 
Ruijsenaars--Schneider (RS) models \cite{Ruij_Schneider, Ruij_CMP1987}. 
Moreover, in the so-called `non-relativistic' limit, they reproduce the 
Ca\-lo\-gero--Moser--Sutherland (CMS) models \cite{Calogero, Sutherland, 
Moser_1975, Olsha_Pere_1976} associated with the $BC$-type root systems. 
However, compared to the translation invariant $A$-type models, the 
geometrical picture underlying the most general classical van Diejen models 
is far less developed. The most probable explanation of this fact is the lack 
of Lax representation for the van Diejen dynamics. For this reason, working 
mainly in a symplectic reduction framework, in the last couple of years we 
undertook the study of the $BC$-type rational van Diejen models 
\cite{Pusztai_NPB2011, Pusztai_NPB2012, Pusztai_NPB2013, Feher_Gorbe_JMP2014, 
Gorbe_Feher_PLA2015, Pusztai_NPB2015}. By going one stage up, in this paper 
we wish to report on our first results about the hyperbolic variants of the 
van Diejen family.

In order to describe the Hamiltonian systems of our interest, let us recall 
that the configuration space of the hyperbolic $n$-particle van Diejen model 
is the open subset
\be
    Q 
    = \{ \lambda = (\lambda_1, \ldots, \lambda_n) \in \bR^n 
        \, | \, 
        \lambda_1 > \ldots > \lambda_n > 0 \}
            \subseteq \bR^n,
\label{Q}
\ee
that can be seen as an open Weyl chamber of type $BC_n$. The cotangent bundle 
of $Q$ is trivial, and it can be naturally identified with the open subset
\be
    P
    = Q \times \bR^n
    = \{ (\lambda, \theta) 
            = (\lambda_1, \ldots, \lambda_n, \theta_1, \ldots, \theta_n) 
                \in \bR^{2 n} 
        \, | \, 
        \lambda_1 > \ldots > \lambda_n > 0 \}
            \subseteq \bR^{2 n}.
\label{P}
\ee
Following the widespread custom, throughout the paper we shall occasionally 
think of the letters $\lambda_a$ and $\theta_a$ $(1 \leq a \leq n)$ as 
globally defined coordinate functions on $P$. For example, using this latter 
interpretation, the canonical symplectic form on the phase space 
$P \cong T^* Q$ can be written as
\be
    \omega = \sum_{c = 1}^n \dd \lambda_c \wedge \dd \theta_c,
\label{symplectic_form}
\ee
whereas the fundamental Poisson brackets take the form
\be
    \{ \lambda_a, \lambda_b \} = 0,
    \quad
    \{ \theta_a, \theta_b \} = 0,
    \quad
    \{ \lambda_a, \theta_b \} = \delta_{a, b}
    \qquad
    (1 \leq a, b \leq n).
\label{PBs}
\ee
The principal goal of this paper is to study the dynamics generated by the 
smooth Hamiltonian function
\be
    H 
    = \sum_{a = 1}^n \cosh(\theta_a)
        \left( 1 + \frac{\sin(\nu)^2}{\sinh(2 \lambda_a)^2} \right)^\half
        \prod_{\substack{c = 1 \\ (c \neq a)}}^n
            \left( 
                1 + \frac{\sin(\mu)^2}{\sinh(\lambda_a - \lambda_c)^2} 
            \right)^\half
            \left( 
                1 + \frac{\sin(\mu)^2}{\sinh(\lambda_a + \lambda_c)^2} 
            \right)^\half,
\label{H}
\ee
where $\mu, \nu \in \bR$ are arbitrary coupling constants satisfying the 
conditions
\be
    \sin(\mu) \neq 0 \neq \sin(\nu).
\label{mu_nu_conds}
\ee
Note that $H$ \eqref{H} does belong to the family of the hyperbolic 
$n$-particle van Diejen Hamiltonians with \emph{two independent parameters} 
$\mu$ and $\nu$ (cf. \eqref{H1_vs_H}). Of course, the values of the parameters 
$\mu$ and $\nu$ really matter only modulo $\pi$.

Now, we briefly outline the content of the paper. In the subsequent 
section, we start with a short overview on some relevant facts and notations 
from Lie theory. Having equipped with the necessary background material, in 
Section \ref{SECTION_Algebraic_properties} we define our Lax matrix \eqref{L} 
for the van Diejen system \eqref{H}, and also investigate its main algebraic 
properties. In Section \ref{SECTION_Dynamics} we turn to the study of the 
Hamiltonian flow generated by \eqref{H}. As the first step, in Theorem 
\ref{THEOREM_Completeness} we formulate the completeness of the corresponding 
Hamiltonian vector field. Most importantly, in Theorem 
\ref{THEOREM_Lax_representation} we provide a Lax representation of the 
dynamics, whereas in Theorem \ref{THEOREM_eigenvalue_dynamics} we establish 
a solution algorithm of purely algebraic nature. Making use of the projection 
method formulated in Theorem \ref{THEOREM_eigenvalue_dynamics}, we also 
initiate the study of the scattering properties of the system \eqref{H}. 
Our rigorous results on the temporal asymptotics of the maximally defined 
trajectories are summarized in Lemma \ref{LEMMA_asymptotics}. Section 
\ref{SECTION_Spectral_invariants} serves essentially two purposes. In 
Subsection \ref{SUBSECTION_link_to_vD} we elaborate the link between our 
special $2$-parameter family of Hamiltonians \eqref{H} and the most general 
$5$-parameter family of hyperbolic van Diejen systems \eqref{H_vD}. At the 
level of the coupling parameters the relationship can be read off from the 
equation \eqref{2parameters}. Furthermore, in Lemma 
\ref{LEMMA_linear_relation} we affirm the equivalence between van Diejen's
commuting family of Hamiltonians and the coefficients of the characteristic 
polynomial of the Lax matrix \eqref{L}. Based on this technical result, in 
Theorem \ref{THEOREM_commuting_eigenvalues} we can infer that the eigenvalues 
of the proposed Lax matrix \eqref{L} provide a commuting family of first 
integrals for the Hamiltonian system \eqref{H}. We conclude the paper with 
Section \ref{SECTION_Discussion}, where we discuss the potential applications, 
and also offer some open problems and conjectures. In particular, in 
\eqref{L_tilde} we propose a Lax matrix for the $3$-parameter family of 
hyperbolic van Diejen systems defined in \eqref{H_1_mu_nu_kappa}.

\section{Preliminaries from group theory}
\label{SECTION_Preliminaries}
This section has two main objectives. Besides fixing the notations used
throughout the paper, we also provide a brief account on some relevant facts
from Lie theory underlying our study of the $2$-parameter family of 
hyperbolic van Diejen systems \eqref{H}. For convenience, our conventions
closely follow Knapp's book \cite{Knapp}.
 
As before, by $n \in \bN = \{1, 2, \dots \}$ we denote the number of particles.
Let $N = 2 n$, and also introduce the shorthand notations
\be 
    \bN_n = \{ 1, \dots, n \}
    \quad \text{and} \quad
    \bN_N = \{1, \dots, N \}.
\label{bN_n} 
\ee
With the aid of the $N \times N$ matrix
\be
    C 
    = \begin{bmatrix}
    	0_n & \bsone_n \\
    	\bsone_n & 0_n
    \end{bmatrix}
\label{C} 
\ee
we define the non-compact real reductive matrix Lie group
\be
	G = U(n, n) = \{ y \in GL(N, \bC) \, | \, y^* C y = C \},
\label{G}
\ee
in which the set of unitary elements
\be
    K = \{ y \in G \, | \, y^* y = \bsone_N \}
    \cong
    U(n) \times U(n)
\label{K}
\ee
forms a maximal compact subgroup. The Lie algebra of $G$ \eqref{G} takes 
the form
\be
	\mfg 
	= \mfu(u, n) 
	= \{ Y \in \mfgl(N, \bC) \, | \, Y^* C + C Y = 0 \},
\label{mfg}
\ee 
whereas for the Lie subalgebra corresponding to $K$ \eqref{K} we have
the identification
\be
    \mfk 
    = \{ Y \in \mfg \, | \, Y^* + Y = 0 \} \cong \mfu(n) \oplus \mfu(n).
\label{mfk}
\ee
Upon introducing the subspace
\be
    \mfp = \{ Y \in \mfg \, | \, Y^* = Y \},
\label{mfp}
\ee
we can write the decomposition $\mfg = \mfk \oplus \mfp$, which is orthogonal 
with respect to the usual trace pairing defined on the matrix Lie algebra 
$\mfg$. Let us note that the restriction of the exponential map onto the 
complementary subspace $\mfp$ \eqref{mfp} is injective. Moreover, the image 
of $\mfp$ under the exponential map can be identified with the set of the 
positive definite elements of the group $U(n, n)$; that is, 
\be
    \exp(\mfp) = \{ y \in U(n, n) \, | \, y > 0 \}.
\label{exp_mfp_identification}
\ee
Notice that, due to the Cartan decomposition $G = \exp(\mfp) K$, the above 
set can be also naturally identified with the non-compact symmetric space 
associated with the pair $(G, K)$, i.e.,
\be
    \exp(\mfp) 
    \cong U(n, n) / (U(n) \times U(n))
    \cong SU(n, n) / S(U(n) \times U(n)).
\label{symm_space}
\ee

To get a more detailed picture about the structure of the reductive Lie 
group $U(n, n)$, in $\mfp$ \eqref{mfp} we introduce the maximal Abelian 
subspace
\be
    \mfa 
    = \{ X = \diag(x_1, \ldots, x_n, -x_1, \ldots, -x_n) 
        \, | \, 
        x_1, \ldots, x_n \in \bR \}. 
\label{mfa}
\ee
Let us recall that we can attain every element of $\mfp$ by conjugating the 
elements of $\mfa$ with the elements of the compact subgroup $K$ \eqref{K}. 
More precisely, the map
\be
    \mfa \times K \ni (X, k) \mapsto k X k^{-1} \in \mfp
\label{mfp_and_mfa_and_K}
\ee
is well-defined and onto. As for the centralizer of $\mfa$ inside $K$ 
\eqref{K}, it turns out to be the Abelian Lie group
\be
    M = Z_K(\mfa)
    = \{ \diag(e^{\ri \chi_1}, \ldots, e^{\ri \chi_n}, 
                e^{\ri \chi_1}, \ldots, e^{\ri \chi_n}) 
        \, | \,
        \chi_1, \ldots, \chi_n \in \bR \}
\label{M}
\ee
with Lie algebra
\be
    \mfm = \{ \diag(\ri \chi_1, \ldots, \ri \chi_n, 
                    \ri \chi_1, \ldots, \ri \chi_n) 
            \, | \,
            \chi_1, \ldots, \chi_n \in \bR \}.
\label{mfm}
\ee
Let $\mfm^\perp$ and $\mfa^\perp$ denote the sets of the off-diagonal 
elements in the subspaces $\mfk$ and $\mfp$, respectively; then clearly
we can write the refined orthogonal decomposition
\be
    \mfg = \mfm \oplus \mfm^\perp \oplus \mfa \oplus \mfa^\perp.
\label{refined_decomposition}
\ee
To put it simple, each Lie algebra element $Y \in \mfg$ can be decomposed as
\be
    Y = Y_\mfm + Y_{\mfm^\perp} + Y_\mfa + Y_{\mfa^\perp}
\label{Y_decomp}
\ee
with unique components belonging to the subspaces indicated by the 
subscripts. 

Throughout our work the commuting family of linear operators
\be
    \ad_X \colon \mfgl(N, \bC) \rightarrow \mfgl(N, \bC),
    \qquad
    Y \mapsto [X, Y]
\label{ad}
\ee
defined for the diagonal matrices $X \in \mfa$ plays a distinguished role. 
Let us note that the (real) subspace 
$\mfm^\perp \oplus \mfa^\perp \subseteq \mfgl(N, \bC)$ is invariant under 
$\ad_X$, whence the restriction
\be
    \wad_X 
    = \ad_X |_{\mfm^\perp \oplus \mfa^\perp}
    \in \mfgl(\mfm^\perp \oplus \mfa^\perp)
\label{wad}
\ee
is a well-defined operator for each 
$X = \diag(x_1, \ldots, x_n, -x_1, \ldots, -x_n) \in \mfa$ with spectrum
\be
    \mathrm{Spec}(\wad_X) 
    = \{x_a - x_b, \pm (x_a + x_b), \pm 2 x_c 
        \, | \, 
        a, b, c \in \bN_n, \, a \neq b\}.
\label{wad_X_spectrum}
\ee
Now, recall that the regular part of the Abelian subalgebra $\mfa$ 
\eqref{mfa} is defined by the subset
\be
    \mfa_\reg 
    = \{ X \in \mfa
        \, | \, 
        \wad_X \text{ is invertible} \},
\label{mfa_reg}
\ee
in which the standard open Weyl chamber
\be
    \mfc 
    = \{ X = \diag(x_1, \ldots, x_n, -x_1, \ldots, -x_n) \in \mfa
        \, | \,
        x_1 > \ldots > x_n > 0 \}
\label{mfc}
\ee
is a connected component. Let us observe that it can be naturally identified 
with the configuration space $Q$ \eqref{Q}; that is, $Q \cong \mfc$. Finally, 
let us recall that the regular part of $\mfp$ \eqref{mfp} is defined as
\be
    \mfp_\reg 
    = \{ k X k^{-1} \in \mfp 
        \, | \, 
        X \in \mfa_\reg \text{ and } k \in K \}.
\label{mfp_reg}
\ee
As a matter of fact, from the map \eqref{mfp_and_mfa_and_K} we can derive
a particularly useful characterization for the open subset 
$\mfp_\reg \subseteq \mfp$. Indeed, the map
\be
    \mfc \times (K / M) \ni (X, k M) \mapsto k X k^{-1} \in \mfp_\reg
\label{mfp_reg_identification}
\ee
turns out to be a diffeomorphism, providing the identification
$\mfp_\reg \cong \mfc \times (K / M)$.

\section{Algebraic properties of the Lax matrix}
\label{SECTION_Algebraic_properties}
Having reviewed the necessary notions and notations from Lie theory, in this 
section we propose a Lax matrix for the hyperbolic van Diejen system of our 
interest \eqref{H}. To make the presentation simpler, with any
$\lambda = (\lambda_1, \dots, \lambda_n) \in \bR^n$ and 
$\theta = (\theta_1, \dots, \theta_n) \in \bR^n$ we associate the real 
$N$-tuples
\be
    \Lambda = (\lambda_1, \dots, \lambda_n, -\lambda_1, \dots, -\lambda_n)
    \quad \text{and} \quad
    \Theta = (\theta_1, \dots, \theta_n, -\theta_1, \dots, -\theta_n),
\label{Lambda_Theta}
\ee
respectively, and also define the $N \times N$ diagonal matrix
\be
    \bsLambda 
    = \diag(\Lambda_1, \dots, \Lambda_N)
    = \diag(\lambda_1, \dots, \lambda_n, -\lambda_1, \dots, -\lambda_n)
    \in \mfa.
\label{bsLambda}
\ee
Notice that if $\lambda \in \bR^n$ is a regular element in the sense that 
the corresponding diagonal matrix $\bsLambda$ \eqref{bsLambda} belongs to 
$\mfa_\reg$ \eqref{mfa_reg}, then for each $j \in \bN_N$ the complex
number
\be
    z_j = - \frac{\sinh(\ri \nu + 2 \Lambda_j)}{\sinh(2 \Lambda_j)}
            \prod_{\substack{c = 1 \\ (c \neq j, j - n)}}^n
                \frac{\sinh(\ri \mu + \Lambda_j - \lambda_c)}
                    {\sinh(\Lambda_j - \lambda_c)}
                \frac{\sinh(\ri \mu + \Lambda_j + \lambda_c)}
                    {\sinh(\Lambda_j + \lambda_c)}
\label{z_a}
\ee
is well-defined. Thinking of $z_j$ as a function of $\lambda$, let us observe 
that its modulus $u_j = \vert z_j \vert$ takes the form
\be
    u_j
    = \left( 1 + \frac{\sin(\nu)^2}{\sinh(2 \Lambda_j)^2} \right)^\half
        \prod_{\substack{c = 1 \\ (c \neq j, j - n)}}^n
            \left( 
                1 + \frac{\sin(\mu)^2}{\sinh(\Lambda_j - \lambda_c)^2} 
            \right)^\half
            \left( 
                1 + \frac{\sin(\mu)^2}{\sinh(\Lambda_j + \lambda_c)^2} 
            \right)^\half,   
\label{u_a}
\ee
and the property $z_{n + a} = \bar{z}_a$ ($a \in \bN_n$) is also clear. 
Next, built upon the functions $z_j$ and $u_j$, we introduce the column 
vector $F \in \bC^N$ with components
\be
    F_a = e^{\frac{\theta_a}{2}} u_a^{\frac{1}{2}}
    \quad \text{and} \quad
    F_{n + a} = e^{-\frac{\theta_a}{2}} \bar{z}_a u_a^{-\frac{1}{2}}
    \qquad
    (a \in \bN_n).
\label{F}
\ee
At this point we are in a position to define our Lax matrix 
$L \in \mfgl(N, \bC)$ with the entries
\be
    L_{k, l} 
    = \frac{\ri \sin(\mu) F_k \bar{F}_l 
        + \ri \sin(\mu - \nu) C_{k, l}}
        {\sinh(\ri \mu + \Lambda_k - \Lambda_l)}
    \qquad
    (k, l \in \bN_N).
\label{L}
\ee
Note that the matrix valued function $L$ is well-defined at each point
$(\lambda, \theta) \in \bR^{N}$ satisfying the regularity condition
$\bsLambda \in \mfa_\reg$. Since $\mfc \subseteq \mfa_\reg$ \eqref{mfc}, 
$L$ makes sense at each point of the phase space $P$ \eqref{P} as well.
To give a motivation for the definition of $L = L(\lambda, \theta; \mu, \nu)$ 
\eqref{L}, let us observe that in its `rational limit' we get back the Lax 
matrix of the rational van Diejen system with two parameters. Indeed, up to 
some irrelevant numerical factors caused by a slightly different convention, 
in the $\alpha \to 0^+$ limit the matrix 
$L(\alpha \lambda, \theta; \alpha \mu, \alpha \nu)$ tends to the rational 
Lax matrix $\cA = \cA(\lambda, \theta; \mu, \nu)$ as defined in the equations 
(4.2)-(4.5) of paper \cite{Pusztai_NPB2011}. In \cite{Pusztai_NPB2011} 
we saw that $\cA$ has many peculiar algebraic properties, that we wish to 
generalize for the proposed hyperbolic Lax matrix $L$ in the rest of this 
section.

\subsection{The matrix $L$ and the Lie group $U(n, n)$}
\label{SUBSECTION_Lax_mat_and_Lie_group}
By inspecting the matrix entries \eqref{L}, it is obvious that $L$ is 
Hermitian. However, it is a less trivial fact that $L$ is closely tied with 
the non-compact Lie group $U(n, n)$ \eqref{G}. The purpose of this subsection 
is to explore this surprising relationship.

\begin{PROPOSITION}
\label{PROPOSITION_L_in_G}
The matrix $L$ \eqref{L} obeys the quadratic equation $L C L = C$. In other 
words, the matrix valued function $L$ takes values in the Lie group $U(n, n)$.
\end{PROPOSITION}

\begin{proof}
Take an arbitrary element $(\lambda, \theta) \in \bR^N$ 
satisfying the regularity condition $\bsLambda \in \mfa_\reg$. We start 
by observing that for each $a \in \bN_n$ the complex conjugates of $z_a$ 
\eqref{z_a} and $F_{n + a}$ \eqref{F} can be obtained by changing the 
sign of the single component $\lambda_a$ of $\lambda$. Therefore, if 
$a, b \in \bN_n$ are arbitrary indices, then by interchanging the components 
$\lambda_a$ and $\lambda_b$ of the $n$-tuple $\lambda$, the expression 
$(L C L)_{a, b} F_{a}^{-1} \bar{F}_{b}^{-1}$ readily transforms into 
$(L C L)_{n + a, n + b} F_{n + a}^{-1} \bar{F}_{n + b}^{-1}$. We capture 
this fact by writing
\be
    \frac{(L C L)_{a, b}}{F_{a} \bar{F}_{b}}
    \underset{\lambda_a \leftrightarrow \lambda_b}{\leadsto}
    \frac{(L C L)_{n + a, n + b}}{F_{n + a} \bar{F}_{n + b}}
    \qquad
    (a, b \in \bN_n).
\label{L1_1}
\ee
Similarly, if $a \neq b$, then from 
$(L C L)_{a, b} F_{a}^{-1} \bar{F}_{b}^{-1}$ we can recover 
$(L C L)_{n + a, b} F_{n + a}^{-1} \bar{F}_{b}^{-1}$
by exchanging $\lambda_a$ for $-\lambda_a$. Schematically, we have
\be
    \frac{(L C L)_{a, b}}{F_{a} \bar{F}_{b}}
    \underset{\lambda_a \leftrightarrow -\lambda_a}{\leadsto}
    \frac{(L C L)_{n + a, b}}{F_{n + a} \bar{F}_{b}}
    \qquad
    (a, b \in \bN_n, \, a \neq b).
\label{L1_2}
\ee
Furthermore, the expression $(L C L)_{a, b} F_{a}^{-1} \bar{F}_{b}^{-1}$ 
reproduces $(L C L)_{a, n + b} F_{a}^{-1} \bar{F}_{n + b}^{-1}$ upon 
swapping $\lambda_b$ for $-\lambda_b$, i.e.,
\be
    \frac{(L C L)_{a, b}}{F_{a} \bar{F}_{b}}
    \underset{\lambda_b \leftrightarrow -\lambda_b}{\leadsto}
    \frac{(L C L)_{a, n + b}}{F_{a} \bar{F}_{n + b}}
    \qquad(a, b \in \bN_n, a \neq b).
\label{L1_3}
\ee
Finally, the relationship between the remaining entries is given by the 
exchange 
\be
    (L C L)_{a, n + a}
    \underset{\lambda_a \leftrightarrow -\lambda_a}{\leadsto}
    (L C L)_{n + a, a}
    \qquad (a \in \bN_n).
\label{L1_4}
\ee
The message of the above equations \eqref{L1_1}-\eqref{L1_4} is quite 
evident. Indeed, in order to prove the desired matrix equation $L C L = C$, 
it does suffice to show that $(L C L)_{a, b} = 0$ for all $a, b \in \bN_n$, 
and also that $(L C L)_{a, n + a} = 1$ for all $a \in \bN_n$.

Recalling the formulae \eqref{F} and \eqref{L}, it is clear that
for all $a \in \bN_n$ we can write
\be
    \frac{(L C L)_{a, a}}{F_a \bar{F}_a}
    = 2 \Real 
    \bigg( 
        \frac{\ri \sin(\mu) z_a + \ri \sin(\mu - \nu)}
            {\sinh(\ri \mu + 2 \lambda_a)}
        -\sum_{\substack{c = 1 \\ (c \neq a)}}^n
            \frac{\sin(\mu)^2 z_c}
                {\sinh(\ri \mu + \lambda_a + \lambda_c) 
                \sinh(\ri \mu - \lambda_a + \lambda_c)}
    \bigg).
\label{LCL_aa}
\ee
To proceed further, we introduce a complex valued function $f_a$ depending 
on a single complex variable $w$ obtained simply by replacing $\lambda_a$ 
with $\lambda_a + w$ in the right-hand side of the above equation 
\eqref{LCL_aa}. Remembering \eqref{z_a}, it is obvious that the resulting 
function is meromorphic with at most first order poles at the points 
 \be
    w \equiv -\lambda_a, \, 
    w \equiv \pm \ri \mu / 2 - \lambda_a, \, 
    w \equiv \Lambda_j - \lambda_a \, (j \in \bN_N)
    \pmod{\ri \pi}.
\label{poles_aa}
\ee
However, by inspecting the terms appearing in the explicit expression
of $f_a$, a straightforward computation reveals immediately that the 
residue of $f_a$ at each of these points is zero, i.e., the singularities 
are in fact removable. As a consequence, $f_a$ can be uniquely extended 
onto the whole complex plane as a periodic entire function with period 
$2 \pi \ri$. Moreover, since $f_a(w)$ vanishes as $\Real(w) \to \infty$, 
the function $f_a$ is clearly bounded. By invoking Liouville's theorem,
we conclude that $f_a(w) = 0$ for all $w \in \bC$, and so
\be
    \frac{(L C L)_{a, a}}{F_a \bar{F}_a} = f_a(0) = 0.
\label{LCL_aa-zero}
\ee

Next, let $a, b \in \bN_n$ be arbitrary indices satisfying $a \neq b$. Keeping 
in mind the definitions \eqref{F} and \eqref{L}, we find at once that
\be
\begin{split}
    \frac{(L C L)_{a, b}}{F_a \bar{F}_b}
    = & \frac{\ri \sin(\mu) \big( \ri \sin(\mu) z_a 
                                    + \ri \sin(\mu - \nu) \big)}
            {\sinh(\ri \mu + \lambda_a - \lambda_b) 
            \sinh(\ri \mu + 2 \lambda_a)}
    + \frac{\ri \sin(\mu) \big( \ri \sin(\mu) \bar{z}_b
                                    + \ri \sin(\mu - \nu) \big)}
            {\sinh(\ri \mu + \lambda_a - \lambda_b)
            \sinh(\ri \mu - 2 \lambda_b)}
    \\
    & + \frac{\ri \sin(\mu) \bar{z}_a}
            {\sinh(\ri \mu - \lambda_a - \lambda_b)}
    + \frac{\ri \sin(\mu) z_b}{\sinh(\ri \mu + \lambda_a + \lambda_b)}
    \\
    & - \sum_{\substack{j = 1 \\ (j \neq a, b, n + a, n + b)}}^N
            \frac{\sin(\mu)^2 z_j}
                {\sinh(\ri \mu + \lambda_a + \Lambda_j)
                \sinh(\ri \mu - \lambda_b + \Lambda_j)}.
\end{split}
\label{LCL_ab}
\ee
Although this equation looks considerably more complicated than 
\eqref{LCL_aa}, it can be analyzed by the same techniques. Indeed, by 
replacing $\lambda_a$ with $\lambda_a + w$ in the right-hand side of 
\eqref{LCL_ab}, we may obtain a meromorphic function $f_{a, b}$ of 
$w \in \bC$ that has at most first order poles at the points
\be
    w \equiv -\lambda_a, \, 
    w \equiv -\ri \mu / 2 - \lambda_a, \, 
    w \equiv -\ri \mu - \lambda_a + \lambda_b, \, 
    w \equiv \Lambda_j - \lambda_a \, (j \in \bN_N)
    \pmod{\ri \pi}.
\label{poles_ab}
\ee
However, the residue of $f_{a, b}$ at each of these points turns out 
to be zero, and $f_{a, b}(w)$ also vanishes as $\Real(w) \to \infty$. Due
to Liouville's theorem we get $f_{a, b}(w) = 0$ for all $w \in \bC$, thus
\be
    \frac{(L C L)_{a, b}}{F_a \bar{F}_b} = f_{a, b}(0) = 0.
\label{LCL_ab-zero}
\ee

Finally, by taking an arbitrary $a \in \bN_n$, from \eqref{F} and \eqref{L}
we see that
\be
    (L C L)_{a, n + a} 
    = u_a^2
        + \frac{(\ri \sin(\mu) z_a + \ri \sin(\mu - \nu))^2}
            {\sinh(\ri \mu + 2 \lambda_a)^2}
        - \sum_{\substack{j = 1 \\ (j \neq a, n + a)}}^N
            \frac{\sin(\mu)^2 z_a z_j}
                {\sinh(\ri \mu + \lambda_a + \Lambda_j)^2}.
\label{LCL_n+a,a}
\ee
By replacing $\lambda_a$ with $\lambda_a + w$ in the right-hand side of 
\eqref{LCL_n+a,a}, we end up with a meromorphic function $f_{n + a}$ of 
the complex variable $w$ that has at most second order poles at the points
\be
    w \equiv -\lambda_a, \, 
    w \equiv -\ri \mu / 2 - \lambda_a, \, 
    w \equiv \Lambda_j - \lambda_a \, (j \in \bN_N)
    \pmod{\ri \pi}.
\label{poles_n+a,a}
\ee
Though the calculations are a bit more involved as in the previous cases, 
one can show that the singularities of $f_{n + a}$ are actually removable. 
Moreover, it is evident that $f_{n + a}(w) \to 1$ as $\Real(w) \to \infty$. 
Liouville's theorem applies again, implying that $f_{n + a}(w) = 1$ for 
all $w \in \bC$. Thus the relationship
\be
    (L C L)_{a, n + a} = f_{n + a}(0) = 1
\label{L_n+a,a-one}
\ee
also follows, whence the proof is complete.
\end{proof}

In the earlier paper \cite{Pusztai_NPB2011} we saw that the rational 
analogue of $L$ \eqref{L} takes values in the symmetric space $\exp(\mfp)$ 
\eqref{symm_space}. We find it reassuring that the proof of Lemma 7 of 
paper \cite{Pusztai_NPB2011} allows a straightforward generalization into 
the present hyperbolic context, too.

\begin{LEMMA}
\label{LEMMA_L_in_exp_p}
At each point of the phase space we have $L \in \exp(\mfp)$.
\end{LEMMA}

\begin{proof}
Recalling the identification \eqref{exp_mfp_identification} and Proposition 
\ref{PROPOSITION_L_in_G}, it is enough to prove that the Hermitian matrix
$L$ \eqref{L} is positive definite. For this reason, take an arbitrary point 
$(\lambda, \theta) \in P$ and keep it fixed. To prove the Lemma, below we 
offer a standard continuity argument by analyzing the dependence of $L$ 
solely on the coupling parameters.

In the very special case when the pair $(\mu, \nu)$ formed by the coupling 
parameters obey the relationship $\sin(\mu - \nu) = 0$, the Lax matrix $L$ 
\eqref{L} becomes a hyperbolic Cauchy-like matrix and the generalized Cauchy 
determinant formula (see e.g. equation (1.2) in \cite{Ruij_CMP1988}) readily 
implies the positivity of all its leading principal minors. Thus, recalling 
Sylvester's criterion, we conclude that $L$ is positive definite.

Turning to the general case, suppose that the pair $(\mu, \nu)$ is restricted 
only by the conditions displayed in \eqref{mu_nu_conds}. It is clear that in 
the $2$-dimensional space of the admissible coupling parameters characterized
by \eqref{mu_nu_conds} one can find a continuous curve with endpoints 
$(\mu, \nu)$ and $(\mu_0, \nu_0)$, where $\mu_0$ and $\nu_0$ satisfy the 
additional requirement $\sin(\mu_0 - \nu_0) = 0$. Since the dependence of the 
Hermitian matrix $L$ on the coupling parameters is smooth, along this curve 
the smallest eigenvalue of $L$ moves continuously. However, it cannot cross 
zero, since by Proposition \ref{PROPOSITION_L_in_G} the matrix $L$ remains 
invertible during this deformation. Therefore, since the eigenvalues of $L$ 
are strictly positive at the endpoint $(\mu_0, \nu_0)$, they must be strictly 
positive at the other endpoint $(\mu, \nu)$ as well.
\end{proof}

\subsection{Commutation relation and regularity}
\label{SUBSECTION_Commut_rel}
As Ruijsenaars has observed in his seminal paper on the translation 
invariant CMS and RS type pure soliton systems, one of the key ingredients 
in their analysis is the fact that their Lax matrices obey certain 
non-trivial commutation relations with some diagonal matrices (for details,
see equation (2.4) and the surrounding ideas in \cite{Ruij_CMP1988}). 
As a momentum map constraint, an analogous commutation relation has also 
played a key role in the geometric study of the rational $C_n$ and $BC_n$ 
RSvD systems (see \cite{Pusztai_NPB2011, Pusztai_NPB2012, 
Feher_Gorbe_JMP2014}). Due to its importance, our first goal in this 
subsection is to set up a Ruijsenaars type commutation relation for the 
proposed Lax matrix $L$ \eqref{L}, too. As a technical remark, we mention 
in passing that from now on we shall apply frequently the standard functional 
calculus on the linear operators $\ad_{\bsLambda}$ \eqref{ad} and 
$\wad_{\bsLambda}$ \eqref{wad} associated with the diagonal matrix 
$\bsLambda \in \mfc$ \eqref{bsLambda}.

\begin{LEMMA}
\label{LEMMA_commut_rel}
The matrix $L$ \eqref{L} and the diagonal matrix $e^{\bsLambda}$ obey the 
Ruijsenaars type commutation relation
\be
    e^{\ri \mu} e^{\ad_{\bsLambda}} L
    - e^{- \ri \mu} e^{-\ad_{\bsLambda}} L
    = 2 \ri \sin(\mu) F F^* + 2 \ri \sin(\mu - \nu) C.
\label{commut_rel} 
\ee
\end{LEMMA}

\begin{proof}
Recalling the matrix entries of $L$, for all $k, l \in \bN_N$ we can 
write that
\be
\begin{split}
    & \left(
        e^{\ri \mu} e^{\ad_{\bsLambda}} L
        - e^{- \ri \mu} e^{-\ad_{\bsLambda}} L
    \right)_{k, l}
    = \left(
        e^{\ri \mu} e^{\bsLambda} L e^{- \bsLambda}
        - e^{- \ri \mu} e^{- \bsLambda} L e^{\bsLambda}
    \right)_{k, l} 
    \\
    & = e^{\ri \mu} e^{\Lambda_k} L_{k, l} e^{- \Lambda_l}
        - e^{- \ri \mu} e^{- \Lambda_k} L_{k, l} e^{\Lambda_l} 
    = 2 \sinh(\ri \mu + \Lambda_k - \Lambda_l) L_{k, l}
    \\
    & = 2 \ri \sin(\mu) F_k \bar{F}_l + 2 \ri \sin(\mu - \nu) C_{k, l} 
    = \left(
        2 \ri \sin(\mu) F F^* + 2 \ri \sin(\mu - \nu) C
    \right)_{k, l},
\end{split}
\label{commut_rel_proof}
\ee
thus \eqref{commut_rel} follows at once. 
\end{proof}

Though the proof of Lemma \ref{LEMMA_commut_rel} is almost trivial, it 
proves to be quite handy in the forthcoming calculations. In particular, 
based on the commutation relation \eqref{commut_rel}, we shall now prove 
that the spectrum of $L$ is \emph{simple}. Heading toward our present goal, 
first let us recall that Lemma \ref{LEMMA_L_in_exp_p} tells us that 
$L \in \exp(\mfp)$. Therefore, as we can infer easily from 
\eqref{mfp_and_mfa_and_K}, one can find some $y \in K$ and a real $n$-tuple 
$\htheta = (\htheta_1, \ldots, \htheta_n) \in \bR^n$ satisfying
\be
    \htheta_1 \geq \ldots \geq \htheta_n \geq 0,
\label{htheta_assumptions}
\ee
such that with the diagonal matrix
\be
    \hbsTheta 
    = \diag(\hTheta_1, \ldots, \hTheta_N)
    = \diag(\htheta_1, \ldots, \htheta_n, 
            -\htheta_1, \ldots, -\htheta_n)
    \in \mfa   
\label{hbsTheta}
\ee
we can write
\be
    L = y e^{2 \hbsTheta} y^{-1}.
\label{L_diagonalized}
\ee
Now, upon defining
\be
    \hat{L} = y^{-1} e^{2 \bsLambda} y \in \exp(\mfp)
    \quad \text{and} \quad
    \hat{F} = e^{-\hbsTheta} y^{-1} e^{\bsLambda} F \in \bC^N,
\label{hat_L_and_hat_F}
\ee
for these new objects we can also set up a commutation relation analogous 
to \eqref{commut_rel}. Indeed, from \eqref{commut_rel} one can derive that 
\be
    e^{\ri \mu} e^{- \hbsTheta} \hat{L} e^{\hbsTheta}
    - e^{- \ri \mu} e^{\hbsTheta} \hat{L} e^{- \hbsTheta}
    = 2 \ri \sin(\mu) \hat{F} \hat{F}^* + 2 \ri \sin(\mu - \nu) C.
\label{hat_commut_rel}
\ee
Componentwise, from \eqref{hat_commut_rel} we conclude that 
\be
    \hat{L}_{k, l} =
    \frac{\ri \sin(\mu) \hat{F}_k \bar{\hat{F}}_l 
            + \ri \sin(\mu - \nu) C_{k, l}}
        {\sinh(\ri \mu - \hTheta_k + \hTheta_l)}
    \qquad
    (k, l \in \bN_N).
\label{hat_L_entries}
\ee
Since $\hat{L}$ \eqref{hat_L_and_hat_F} is a positive definite matrix, 
its diagonal entries are strictly positive. Therefore, by exploiting
\eqref{hat_L_entries}, we can write
\be
    0 < \hat{L}_{k, k} = \vert \hat{F}_k \vert^2.
\label{hat_F_comp}
\ee
The upshot of this trivial observation is that $\hat{F}_k \neq 0$ for all 
$k \in \bN_N$.

To proceed further, notice that for the inverse matrix 
$\hat{L}^{-1} = C \hat{L} C$ we can also cook up an equation analogous to 
\eqref{hat_commut_rel}. Indeed, by simply multiplying both sides of 
\eqref{hat_commut_rel} with the matrix $C$ \eqref{C}, we obtain
\be
    e^{\ri \mu} e^{\hbsTheta} \hat{L}^{-1} e^{- \hbsTheta}
    - e^{- \ri \mu} e^{- \hbsTheta} \hat{L}^{-1} e^{\hbsTheta}
    = 2 \ri \sin(\mu) (C \hat{F}) (C \hat{F})^* + 2 \ri \sin(\mu - \nu) C,
\label{hat_commut_rel_inv}
\ee
that leads immediately to the matrix entries
\be
    (\hat{L}^{-1})_{k, l} =
    \frac{\ri \sin(\mu) (C \hat{F})_k \overline{(C \hat{F})}_l 
            + \ri \sin(\mu - \nu) C_{k, l}}
        {\sinh(\ri \mu + \hTheta_k - \hTheta_l)}
    \qquad
    (k, l \in \bN_N).
\label{hat_L_inv_entries}
\ee
For further reference, we now spell out the trivial equation
\be 
    \delta_{k, l} = \sum_{j = 1}^N \hat{L}_{k, j} (\hat{L}^{-1})_{j, l}
\label{hat_trivial}
\ee
for certain values of $k, l \in \bN_N$. First, by plugging the explicit 
formulae \eqref{hat_L_entries} and \eqref{hat_L_inv_entries} into the 
relationship \eqref{hat_trivial}, with the special choice of indices 
$k = l = a \in \bN_n$ one finds that
\be
\begin{split}
    0 
    = & 1 + \frac{\sin(\mu - \nu)^2}{\sinh(\ri \mu - 2 \htheta_a)^2}
        + \frac{2 \sin(\mu) \sin(\mu - \nu) 
                \hat{F}_a \bar{\hat{F}}_{n + a}}
            {\sinh(\ri \mu - 2 \htheta_a)^2} \\
    & + \sin(\mu)^2 \hat{F}_a \bar{\hat{F}}_{n + a}
        \sum_{c = 1}^n 
            \left(
                \frac{\bar{\hat{F}}_c \hat{F}_{n + c}}
                    {\sinh(\ri \mu - \htheta_a + \htheta_c)^2}
                + \frac{\hat{F}_c \bar{\hat{F}}_{n + c}}
                    {\sinh(\ri \mu - \htheta_a - \htheta_c)^2}
            \right).
\end{split}
\label{rel_1}
\ee
Second, if $k = a$ and $l = n + a$ with some $a \in \bN_n$, then from 
\eqref{hat_trivial} we obtain
\be
\begin{split}
    & \sin(\mu)^2 
        \sum_{c = 1}^n 
            \left(
                \frac{\bar{\hat{F}}_c \hat{F}_{n + c}}
                    {\sinh(\ri \mu - \htheta_a + \htheta_c)
                    \sinh(\ri \mu + \htheta_c + \htheta_a)}
                + \frac{\hat{F}_c \bar{\hat{F}}_{n + c}}
                    {\sinh(\ri \mu - \htheta_a - \htheta_c)
                    \sinh(\ri \mu - \htheta_c + \htheta_a)}
            \right) 
    \\
    & \quad
    = \ri \sin(\mu - \nu) 
    \left(
        \frac{1}{\sinh(\ri \mu - 2 \htheta_a)}
        + \frac{1}{\sinh(\ri \mu + 2 \htheta_a)}
    \right). 
\end{split}
\label{rel_2}
\ee
Third, if $k = a$ and $l = b$ with some $a, b \in \bN_n$ satisfying 
$a \neq b$, then the relationship \eqref{hat_trivial} immediately leads to 
the equation
\be
\begin{split}
    & \sin(\mu)^2 
        \sum_{c = 1}^n 
            \left(
                \frac{\bar{\hat{F}}_c \hat{F}_{n + c}}
                    {\sinh(\ri \mu - \htheta_a + \htheta_c)
                    \sinh(\ri \mu + \htheta_c - \htheta_b)}
                + \frac{\hat{F}_c \bar{\hat{F}}_{n + c}}
                    {\sinh(\ri \mu - \htheta_a - \htheta_c)
                    \sinh(\ri \mu - \htheta_c - \htheta_b)}
            \right) 
    \\
    & \quad
    = - \frac{\sin(\mu) \sin(\mu - \nu)}
        {\sinh(\ri \mu - \htheta_a - \htheta_b)}  
    \left(
        \frac{1}{\sinh(\ri \mu - 2 \htheta_a)}
        + \frac{1}{\sinh(\ri \mu - 2 \htheta_b)}
    \right). 
\end{split}
\label{rel_3}
\ee
At this point we wish to emphasize that during the derivation of the 
last two equations \eqref{rel_2} and \eqref{rel_3} it proves to be 
essential that each component of the column vector $\hat{F}$ 
\eqref{hat_L_and_hat_F} is nonzero, as we have seen in \eqref{hat_F_comp}.

\begin{LEMMA}
\label{LEMMA_regularity}
Under the additional assumption on the coupling parameters
\be
    \sin(2 \mu - \nu) \neq 0,
\label{coupling_cond_extra}
\ee
the spectrum of the matrix $L$ \eqref{L} is simple of the form
\be
    \mathrm{Spec}(L) 
    = \{ e^{\pm 2 \htheta_a} \, | \, a \in \bN_n \},
\label{spec_L}
\ee
where $\htheta_1 > \ldots > \htheta_n > 0$. In other words, $L$ is regular 
in the sense that $L \in \exp(\mfp_\reg)$.
\end{LEMMA}

\begin{proof}
First, let us \emph{suppose} that $\htheta_a = 0$ for some $a \in \bN_n$. 
With this particular index $a$, from equation \eqref{rel_1} we infer that
\be
\begin{split}
    0 
    = & 1 - \frac{\sin(\mu - \nu)^2}{\sin(\mu)^2}
        - \frac{2 \sin(\mu - \nu) 
                \hat{F}_a \bar{\hat{F}}_{n + a}}
            {\sin(\mu)} \\
    & + \sin(\mu)^2 \hat{F}_a \bar{\hat{F}}_{n + a}
        \sum_{c = 1}^n 
            \left(
                \frac{\bar{\hat{F}}_c \hat{F}_{n + c}}
                    {\sinh(\ri \mu + \htheta_c)^2}
                + \frac{\hat{F}_c \bar{\hat{F}}_{n + c}}
                    {\sinh(\ri \mu - \htheta_c)^2}
            \right),
\end{split}
\label{L2_rel_1}
\ee
while \eqref{rel_2} leads to the relationship
\be
    \sin(\mu)^2 
        \sum_{c = 1}^n 
            \left(
                \frac{\bar{\hat{F}}_c \hat{F}_{n + c}}
                    {\sinh(\ri \mu + \htheta_c)^2}
                + \frac{\hat{F}_c \bar{\hat{F}}_{n + c}}
                    {\sinh(\ri \mu - \htheta_c)^2}
            \right)
    = \frac{2 \sin(\mu - \nu)}{\sin(\mu)}.
\label{L2_rel_2}
\ee
Now, by plugging \eqref{L2_rel_2} into \eqref{L2_rel_1}, we obtain
\be
    0 
    = 1 - \frac{\sin(\mu - \nu)^2}{\sin(\mu)^2}
    = \frac{\sin(\mu)^2 - \sin(\mu - \nu)^2 }{\sin(\mu)^2}
    = \frac{\sin(\nu) \sin(2 \mu - \nu)}{\sin(\mu)^2},
\label{trig_contradiction_1}
\ee
which clearly contradicts the assumptions imposed in the equations 
\eqref{mu_nu_conds} and \eqref{coupling_cond_extra}. Thus, we are 
forced to conclude that for all $a \in \bN_n$ we have $\htheta_a \neq 0$.

Second, let us \emph{suppose} that $\htheta_a = \htheta_b$ for some 
$a, b \in \bN_n$ satisfying $a \neq b$. With these particular indices 
$a$ and $b$, equation \eqref{rel_3} takes the form
\be
    \sin(\mu)^2 
        \sum_{c = 1}^n 
            \left(
                \frac{\bar{\hat{F}}_c \hat{F}_{n + c}}
                    {\sinh(\ri \mu - \htheta_a + \htheta_c)^2}
                + \frac{\hat{F}_c \bar{\hat{F}}_{n + c}}
                    {\sinh(\ri \mu - \htheta_a - \htheta_c)^2}
            \right) 
    = - \frac{2 \sin(\mu) \sin(\mu - \nu)}{\sinh(\ri \mu - 2 \htheta_a)^2}.
\label{L2_rel_3}
\ee
Now, by plugging this formula into \eqref{rel_1}, we obtain immediately that
\be
    0 = 1 + \frac{\sin(\mu - \nu)^2}{\sinh(\ri \mu - 2 \htheta_a)^2},
\label{L2_rel_4}
\ee
which in turn implies that
\be
\begin{split}
    & \sin(\mu - \nu)^2 
    = - \sinh(\ri \mu - 2 \htheta_a)^2 \\
    & \quad
    = \sin(\mu)^2 \cosh(2 \htheta_a)^2 
        - \cos(\mu)^2 \sinh(2 \htheta_a)^2
        + \ri \sin(\mu) \cos(\mu) \sinh(4 \htheta_a).
\end{split}
\label{L2_ab_key}
\ee
Since $\htheta_a \neq 0$ and since $\sin(\mu) \neq 0$, the imaginary part 
of the above equation leads to the relation $\cos(\mu) = 0$, whence 
$\sin(\mu)^2 = 1$ also follows. Now, by plugging these observations into 
the real part of \eqref{L2_ab_key}, we end up with the contradiction
\be
    1 \geq \sin(\mu - \nu)^2 = \cosh(2 \htheta_a)^2 > 1.
\label{L2_contradiction}
\ee
Thus, if $a, b \in \bN_n$ and $a \neq b$, then necessarily we have 
$\htheta_a \neq \htheta_b$.
\end{proof}

Since the spectrum of $L$ \eqref{L} is simple, it follows that the
dependence of the eigenvalues on the matrix entries is smooth. Therefore,
recalling \eqref{spec_L}, it is clear that each $\htheta_c$ $(c \in \bN_n)$
can be seen as a smooth function on $P$ \eqref{P}, i.e.,
\be
    \htheta_c \in C^\infty(P).
\label{htheta_smooth}
\ee
To conclude this subsection, we also offer a few remarks on the additional
constraint appearing in \eqref{coupling_cond_extra}, that we keep in
effect in the rest of the paper. Naively, this assumption excludes a 
$1$-dimensional subset from the $2$-dimensional space of the parameters 
$(\mu, \nu)$. However, looking back to the Hamiltonian $H$ \eqref{H},
it is clear that the effective coupling constants of our van Diejen systems
are rather the positive numbers $\sin(\mu)^2$ and $\sin(\nu)^2$. Therefore, 
keeping in mind \eqref{mu_nu_conds}, on the parameters $\mu$ and $\nu$ we 
could have imposed the requirement, say,
\be
    (\mu, \nu) 
    \in \left( (0, \pi / 4) \times [ -\pi / 2, 0 ) \right)
        \cup \left( [\pi / 4, \pi / 2] \times ( 0, \pi / 2] \right),
\label{mu_nu_spec}
\ee 
at the outset. The point is that, under the requirement \eqref{mu_nu_spec}, 
the equation $\sin(2 \mu - \nu) = 0$ is equivalent to the pair of equations 
$\sin(\mu)^2 = 1 / 2$ and $\sin(\nu)^2 = 1$. To put it differently, our 
observation is that under the assumptions \eqref{mu_nu_conds} and 
\eqref{coupling_cond_extra} the pair $(\sin(\mu)^2, \sin(\nu)^2)$ formed by 
the relevant coupling constants can take on any values from the `square' 
$(0, 1] \times (0, 1]$, except the \emph{single point} $(1 / 2, 1)$.
From the proof of Lemma \ref{LEMMA_regularity}, especially from equation 
\eqref{L2_rel_2}, one may get the impression that even this very slight 
technical assumption can be relaxed by further analyzing the properties of 
column vector $\hat{F}$ \eqref{hat_L_and_hat_F}. However, we do not wish to 
pursue this direction in the present paper.

\section{Analyzing the dynamics}
\label{SECTION_Dynamics}
In this section we wish to study the dynamics generated by the Hamiltonian
$H$ \eqref{H}. Recalling the formulae \eqref{u_a} and \eqref{L}, by the
obvious relationship
\be
    H = \sum_{c = 1}^n \cosh(\theta_c) u_c = \half \tr(L)
\label{H_and_u_and_L}
\ee
we can make the first contact of our van Diejen system with the proposed Lax 
matrix $L$ \eqref{L}. As an important ingredient of the forthcoming analysis,
let us introduce the Hamiltonian vector field $\bsX_H \in \mfX(P)$ with the
usual definition
\be
    \bsX_H [f] = \{ f, H \}
    \qquad
    (f \in C^\infty(P)).
\label{bsX_H}
\ee
Working with the convention \eqref{PBs}, for the time evolution of the 
global coordinate functions $\lambda_a$ and $\theta_a$ $(a \in \bN_n)$ 
we can clearly write
\begin{align}
    & \dot{\lambda}_a 
        = \bsX_H [\lambda_a]
        = \PD{H}{\theta_a}
        = \sinh(\theta_a) u_a,
    \label{lambda_dot}
    \\
    & \dot{\theta}_a 
        = \bsX_H [\theta_a]
        = -\PD{H}{\lambda_a}
        = - \sum_{c = 1}^n \cosh(\theta_c) u_c \PD{\ln(u_c)}{\lambda_a}.
    \label{theta_dot}
\end{align}
To make the right-hand side of \eqref{theta_dot} more explicit, let us 
display the logarithmic derivatives of the constituent functions $u_c$. 
Notice that for all $a \in \bN_n$ we can write
\be
    \PD{\ln(u_a)}{\lambda_a}
    = -\Real 
        \bigg(
            \frac{2 \ri \sin(\nu)}
                {\sinh(2 \lambda_a) \sinh(\ri \nu + 2 \lambda_a)}
            +\sum_{\substack{j = 1 \\ (j \neq a, n + a)}}^N
                \frac{\ri \sin(\mu)}
                    {\sinh(\lambda_a - \Lambda_j) 
                    \sinh(\ri \mu + \lambda_a - \Lambda_j)}
        \bigg),
\label{PD_u_a_lambda_a}
\ee
while if $c \in \bN_n$ and $c \neq a$, then we have
\be
    \PD{\ln(u_c)}{\lambda_a}
    = \Real
        \bigg(
            \frac{\ri \sin(\mu)}
                {\sinh(\lambda_a - \lambda_c) 
                \sinh(\ri \mu + \lambda_a - \lambda_c)}
            -\frac{\ri \sin(\mu)}
                {\sinh(\lambda_a + \lambda_c) 
                \sinh(\ri \mu + \lambda_a + \lambda_c)}
        \bigg).
\label{PD_u_c_lambda_a}
\ee
The rest of this section is devoted to the study of the Hamiltonian 
dynamical system \eqref{lambda_dot}-\eqref{theta_dot}.

\subsection{Completeness of the Hamiltonian vector field}
\label{SUBSECTION_Completeness}
Undoubtedly, the Hamiltonian \eqref{H} does not take the usual form one finds
in the standard textbooks on classical mechanics. It is thus inevitable that
we have even less intuition about the generated dynamics than in the case of
the `natural systems' characterized by a kinetic term plus a potential. To
get a finer picture about the solutions of the Hamiltonian dynamics 
\eqref{lambda_dot}-\eqref{theta_dot}, we start our study with a brief analysis
on the completeness of the Hamiltonian vector field $\bsX_H$ \eqref{bsX_H}.

As the first step, we introduce the strictly positive constant
\be
    \cS = \min \{ \vert \sin(\mu) \vert, \vert \sin(\nu) \vert \} \in (0, 1].
\label{cS}
\ee
Giving a glance at \eqref{u_a}, it is evident that
\be
    u_n 
    > \left( 1 + \frac{\sin(\nu)^2}{\sinh(2 \lambda_n)^2} \right)^\half
    > \frac{\vert \sin(\nu) \vert}{\sinh(2 \lambda_n)}
    \geq \frac{\cS}{\sinh(2 \lambda_n)},
\label{u_n_estim}
\ee
while for all $c \in \bN_{n - 1}$ we can write
\be
    u_c 
    > \left( 
        1 + \frac{\sin(\mu)^2}{\sinh(\lambda_c - \lambda_{c + 1})^2} 
    \right)^\half
    > \frac{\vert \sin(\mu) \vert}{\sinh(\lambda_c - \lambda_{c + 1})}
    \geq \frac{\cS}{\sinh(\lambda_c - \lambda_{c + 1})}.
\label{u_c_estim}
\ee
Keeping in mind the above trivial inequalities, we are ready to prove the 
following result.

\begin{THEOREM}
\label{THEOREM_Completeness}
The Hamiltonian vector field $\bsX_H$ \eqref{bsX_H} generated by the van Diejen
type Hamiltonian function $H$ \eqref{H} is complete. That is, the maximum
interval of existence of each integral curve of $\bsX_H$ is the whole real 
axis $\bR$.
\end{THEOREM}

\begin{proof}
Take an arbitrary point
\be
    \gamma_0 = (\lambda^{(0)}, \theta^{(0)}) \in P,
\label{gamma_0}
\ee
and let
\be
    \gamma \colon (\alpha, \beta) \to P,
    \qquad
    t \mapsto \gamma(t) = (\lambda(t), \theta(t))
\label{gamma}
\ee
be the unique maximally defined integral curve of $\bsX_H$ with 
$-\infty \leq \alpha < 0 < \beta \leq \infty$ satisfying the initial condition
$\gamma(0) = \gamma_0$. Since the Hamiltonian $H$ is smooth, the existence,
the uniqueness, and also the smoothness of such a maximal solution are obvious.
Our goal is to show that for the domain of the  maximally defined trajectory
$\gamma$ \eqref{gamma} we have  $(\alpha, \beta) = \bR$; that is, 
$\alpha = -\infty$ and $\beta = \infty$.

Arguing by contradiction, first let us \emph{suppose} that $\beta < \infty$. 
Since the Hamiltonian $H$ is a first integral, for all $t \in (\alpha, \beta)$ 
and for all $a \in \bN_n$ we can write
\be
    H(\gamma_0) 
    = H(\gamma(t)) 
    = \sum_{c = 1}^n \cosh(\theta_c(t)) u_c(\lambda(t))
    > \cosh(\theta_a(t)) u_a(\lambda(t)),
\label{estim_1}
\ee
whence the estimation
\be
    H(\gamma_0) 
    > \cosh(\vert \theta_a(t) \vert) \geq \half e^{\vert \theta_a(t) \vert}
\label{estim_2}
\ee
is also immediate. Thus, upon introducing the cube
\be
    \cC
    = [-\ln(2 H(\gamma_0)), \ln(2 H(\gamma_0))]^n 
        \subseteq \bR^n,
\label{cube}
\ee
from \eqref{estim_2} we infer at once that
\be
    \theta(t) \in \cC
    \qquad
    (t \in (\alpha, \beta)).
\label{theta_in_cube}
\ee
Turning to the equations \eqref{lambda_dot} and \eqref{estim_1}, we can 
cook up an estimation on the growing of the vector $\lambda(t)$, too. 
Indeed, we see that
\be
    \vert \dot{\lambda}_1(t) \vert
    = \sinh(\vert \theta_1(t) \vert) u_1(\lambda(t))
    \leq \cosh(\vert \theta_1(t) \vert) u_1(\lambda(t))
    < H(\gamma_0)
    \qquad
    (t \in (\alpha, \beta)),
\label{lambda_dot_estim}
\ee 
that implies immediately that for all $t \in [0, \beta)$ we have
\be
    \vert \lambda_1(t) - \lambda_1^{(0)} \vert
    = \vert \lambda_1(t) - \lambda_1(0) \vert
    = \left \vert \int_0^t \dot{\lambda}_1(s) \, \dd s \right \vert
    \leq \int_0^t \vert \dot{\lambda}_1(s) \vert \, \dd s
    \leq t H(\gamma_0) < \beta H(\gamma_0).
\label{lambda_1_estim}
\ee
Therefore, with the aid of the strictly positive constant 
\be
    \rho = \lambda_1^{(0)} + \beta H(\gamma_0) \in (0, \infty), 
\label{rho}
\ee    
we end up with the estimation
\be
    \lambda_1(t) 
    = \vert \lambda_1(t) \vert
    = \vert \lambda_1^{(0)} + \lambda_1(t) - \lambda_1^{(0)} \vert
    \leq \vert \lambda_1^{(0)} \vert 
        + \vert \lambda_1(t) - \lambda_1^{(0)} \vert
    < \rho
    \qquad
    (t \in [0, \beta)).
\label{lambda_1_estim_OK}
\ee
Since $\lambda(t)$ moves in the configuration space $Q$ \eqref{Q}, the
above observation entails that
\be
    \rho > \lambda_1(t) > \ldots > \lambda_n(t) > 0
    \qquad
    (t \in [0, \beta)).
\label{lambda_estim}
\ee

To proceed further, now for all $\eps > 0$ we define the subset 
$Q_\eps \subseteq \bR^n$ consisting of those real $n$-tuples 
$x = (x_1, \dots, x_n) \in \bR^n$ that satisfy the inequalities
\be
    \rho \geq x_1
    \text{ and }
    2 x_n \geq \eps
    \text{ and }
    x_c \geq x_{c + 1} + \eps
    \text{ for all } c \in \bN_{n - 1},
\label{Q_eps_def}
\ee
simultaneously. In other words, 
\be
    Q_\eps 
    = \{ x \in \bR^n \, | \, \rho \geq x_1 \} 
        \cap \{ x \in \bR^n \, | \, 2 x_n \geq \eps \}
        \cap
            \bigcap_{c = 1}^{n - 1} 
                \{ x \in \bR^n \, | \, x_c - x_{c + 1} \geq \eps \}.
\label{Q_eps_OK}
\ee
Notice that $Q_\eps$ is a bounded and closed subset of $\bR^n$. Moreover, by
comparing the definitions \eqref{Q} and \eqref{Q_eps_def}, it is evident that
$Q_\eps \subseteq Q$. Since the cube $\cC$ \eqref{cube} is also a compact 
subset of $\bR^n$, we conclude that the Cartesian product $Q_\eps \times \cC$ 
is a \emph{compact} subset of the phase space $P$ \eqref{P}. Therefore, due 
to the assumption $\beta < \infty$, after some time the maximally defined
trajectory $\gamma$ \eqref{gamma} escapes from $Q_\eps \times \cC$, as can
be read off from any standard reference on dynamical systems (see e.g.
Theorem 2.1.18 in \cite{AM}). More precisely, there is some 
$\tau_\eps \in [0,\beta)$ such that
\be
    (\lambda(t), \theta(t)) 
    \in 
    P \setminus (Q_\eps \times \cC)
    = \left( (Q \setminus Q_\eps) \times \cC \right)
        \cup
        \left( Q \times (\bR^n \setminus \cC) \right)    
    \qquad
    (t \in (\tau_\eps, \beta)),
\label{gamma_escapes}
\ee
where the union above is actually a disjoint union. For instance,
due to the relationship \eqref{theta_in_cube}, at the mid-point
\be
    t_\eps = \frac{\tau_\eps + \beta}{2} \in (\tau_\eps, \beta)
\label{t_eps}
\ee
we can write that
\be
    \lambda(t_\eps) 
    \in 
    Q \setminus Q_\eps \subseteq \bR^n \setminus Q_\eps.
\label{lambda_escapes}
\ee 
Therefore, simply by taking the complement of $Q_\eps$ \eqref{Q_eps_OK},
and also keeping in mind \eqref{lambda_1_estim_OK}, it is evident that 
\be
    \min 
    \{ \lambda_1(t_\eps) - \lambda_2(t_\eps),
        \dots,
        \lambda_{n - 1}(t_\eps) - \lambda_n(t_\eps),
        2 \lambda_n(t_\eps)
    \} < \eps,
\label{key_for_completeness}
\ee
which in turn implies that
\be
    \max 
    \left\{ 
        \frac{1}{\sinh(\lambda_1(t_\eps) - \lambda_2(t_\eps))},
        \dots,
        \frac{1}{\sinh(\lambda_{n - 1}(t_\eps) - \lambda_n(t_\eps))},
        \frac{1}{\sinh(2 \lambda_n(t_\eps))}
    \right\} 
    > \frac{1}{\sinh(\eps)}.
\label{key_for_completeness_OK}
\ee
Now, since $\eps > 0$ was arbitrary, the estimations \eqref{u_n_estim} and 
\eqref{u_c_estim} immediately lead to the contradiction
\be
\begin{split}
    H(\gamma_0) 
    & = H(\gamma(t_\eps)) 
    = \sum_{c = 1}^n \cosh(\theta_c(t_\eps)) u_c(\lambda(t_\eps)) 
    \\
    & \geq \sum_{c = 1}^n u_c(\lambda(t_\eps))
    > \frac{\cS}{\sinh(2 \lambda_n(t_\eps))}
        + \sum_{c = 1}^{n - 1} 
            \frac{\cS}{\sinh(\lambda_c(t_\eps) - \lambda_{c + 1}(t_\eps))}
    > \frac{\cS}{\sinh(\eps)}.
\end{split}
\label{contradiction}
\ee
Therefore, necessarily, $\beta = \infty$.

Either by repeating the above ideas, or by invoking a time-reversal argument, 
one can also show that $\alpha = -\infty$, whence the proof is complete.
\end{proof}

\subsection{Dynamics of the vector $F$}
\label{SUBSECTION_Evolution_of_F}
Looking back to the definition \eqref{L}, we see that the column vector $F$ 
\eqref{F} is important building block of the matrix $L$. Therefore, the study 
of the derivative of $L$ along the Hamiltonian vector field $\bsX_H$ 
\eqref{bsX_H} does require close control over the derivative of the components 
of $F$, too. Upon introducing the auxiliary functions
\be
    \varphi_k = \frac{1}{F_k} \bsX_H [F_k]
    \qquad
    (k \in \bN_N), 
\label{varphi}
\ee
for all $a \in \bN_n$ we can write
\be
\begin{split}
    2 \varphi_a  
    & = \bsX_H [\ln (F_a^2)]
    = \bsX_H [\theta_a + \ln(u_a)]
    = \left\{ \theta_a + \ln(u_a), H \right\} 
    \\
    & = \sum_{c = 1}^n
        \left(
            \sinh(\theta_c) u_c \PD{\ln(u_a)}{\lambda_c}
            - \cosh(\theta_c) u_c \PD{\ln(u_c)}{\lambda_a}
        \right).
\end{split}
\label{varphi_a}
\ee
Therefore, due to the explicit formulae \eqref{PD_u_a_lambda_a} and 
\eqref{PD_u_c_lambda_a}, we have complete control over the first $n$
components of \eqref{varphi}. Turning to the remaining components, from 
the definition \eqref{F} it is evident that 
$F_{n + a} = F_a^{-1} \bar{z}_a$, whence the relationship
\be
    \varphi_{n + a}
    = -\varphi_a + \frac{1}{\bar{z}_a} \bsX_H [\bar{z}_a]
    = -\varphi_a 
        + \sum_{c = 1}^n 
            \sinh(\theta_c) u_c \frac{1}{\bar{z}_a} \PD{\bar{z}_a}{\lambda_c}
\label{varphi_n+a}
\ee
follows immediately. Notice that for all $a \in \bN_n$ we can write that
\be
    \frac{1}{z_a} \PD{z_a}{\lambda_a}
    = -\frac{2 \ri \sin(\nu)}
            {\sinh(2 \lambda_a) \sinh(\ri \nu + 2 \lambda_a)}
        -\sum_{\substack{j = 1 \\ (j \neq a, n + a)}}^N
            \frac{\ri \sin(\mu)}
                {\sinh(\lambda_a - \Lambda_j) 
                \sinh(\ri \mu + \lambda_a -\Lambda_j)},
\label{PD_z_a_lambda_a}
\ee
whereas if $c \in \bN_n$ and $c \neq a$, then we find immediately that
\be
    \frac{1}{z_a} \PD{z_a}{\lambda_c}
    = \frac{\ri \sin(\mu)}
            {\sinh(\lambda_a - \lambda_c) 
            \sinh(\ri \mu + \lambda_a - \lambda_c)}
        -\frac{\ri \sin(\mu)}
            {\sinh(\lambda_a + \lambda_c) 
            \sinh(\ri \mu + \lambda_a + \lambda_c)}.
\label{PD_z_a_lambda_c}
\ee
The above observations can be summarized as follows.

\begin{PROPOSITION}
\label{PROPOSITION_varphi}
For the derivative of the components of the function $F$ \eqref{F} along the 
Hamiltonian vector field $\bsX_H$ \eqref{bsX_H} we have
\be
    \bsX_H [F_k] = \varphi_k F_k
    \qquad
    (k \in \bN_N),
\label{F_der}
\ee
where for each $a \in \bN_n$ we can write
\be
    \varphi_a 
    = \Real
        \bigg(
            \frac{\ri \sin(\nu) e^{-\theta_a} u_a}
                {\sinh(2 \lambda_a) \sinh(\ri \nu + 2 \lambda_a)}
            + \half \sum_{\substack{j = 1 \\ (j \neq a, n + a)}}^N
                \frac{\ri \sin(\mu) (e^{-\theta_a} u_a + e^{\Theta_j} u_j)}
                    {\sinh(\lambda_a - \Lambda_j) 
                    \sinh(\ri \mu + \lambda_a - \Lambda_j)}
        \bigg),
\label{varphi_a_OK}
\ee
whereas
\be
    \varphi_{n + a} 
    = -\varphi_a
    - \frac{2 \ri \sin(\nu) \sinh(\theta_a) u_a}
        {\sinh(2 \lambda_a) \sinh(\ri \nu - 2 \lambda_a)}
    -\sum_{\substack{j = 1 \\ (j \neq a, n + a)}}^N
        \frac{\ri \sin(\mu) (\sinh(\theta_a) u_a - \sinh(\Theta_j) u_j)}
            {\sinh(\lambda_a - \Lambda_j) 
            \sinh(\ri \mu - \lambda_a + \Lambda_j)}.
\label{varphi_n+a_OK}
\ee
\end{PROPOSITION}

By invoking Proposition \ref{PROPOSITION_L_in_G}, let us observe that for the 
inverse of the matrix $L$ \eqref{L} we can write that $L^{-1} = C L C$, whence
for the Hermitian matrix $L - L^{-1}$ we have
\be
    (L - L^{-1}) C + C (L - L^{-1})
    = L C - C L C^2 + C L - C^2 L C 
    = 0.
\label{L-L_inv_Hermitian}
\ee
Thus, the matrix valued smooth function $(L - L^{-1}) / 2$ defined on the
phase space $P$ \eqref{P} takes values in the subspace $\mfp$ \eqref{mfp}.
Therefore, by taking its projection onto the Abelian subspace $\mfa$ 
\eqref{mfa}, we obtain the diagonal matrix
\be
    D = (L - L^{-1})_\mfa / 2 \in \mfa
\label{D}
\ee
with diagonal entries
\be
    D_{j, j} = \sinh(\Theta_j) u_j
    \qquad
    (j \in \bN_N).
\label{D_entries}
\ee
Next, by projecting the function $(L - L^{-1}) / 2$ onto the complementary
subspace $\mfa^\perp$, we obtain the off-diagonal matrix
\be
    Y=(L - L^{-1})_{\mfa^\perp} / 2 \in \mfa^\perp,
\label{Y}
\ee
which in turn allows us to introduce the matrix valued smooth function
\be
    Z = \sinh(\wad_{\bsLambda})^{-1} Y \in \mfm^\perp,
\label{Z}
\ee
too. Since $\lambda \in Q$ \eqref{Q}, the corresponding diagonal matrix 
$\bsLambda$ \eqref{bsLambda} is regular in the sense that it takes values 
in the open Weyl chamber $\mfc \subseteq \mfa_\reg$ \eqref{mfc}. Therefore, 
$Z$ is indeed a well-defined off-diagonal $N \times N$ matrix, and its 
non-trivial entries take the form
\be
    Z_{k, l} 
    = \frac{Y_{k, l}}{\sinh(\Lambda_k - \Lambda_l)}
    = \frac{L_{k, l} - (L^{-1})_{k, l}}
        {2 \sinh(\Lambda_k - \Lambda_l)}
    \qquad
    (k, l \in \bN_N, \, k \neq l).
\label{Z_entries}
\ee
Utilizing $Z$, for each $a \in \bN_n$ we also define the function
\be
    \cM_a 
    = \frac{\ri}{F_a} \Imag((Z F)_a) 
    = \frac{\ri}{F_a} 
        \Imag \left( \sum_{j = 1}^N Z_{a, j} F_j \right)
        \in C^\infty(P).
\label{cM_a}
\ee
Recalling the subspace $\mfm$ \eqref{mfm}, it is clear that
\be
    B_\mfm = \diag(\cM_1, \dots, \cM_n, \cM_1, \dots, \cM_n) \in \mfm
\label{B_mfm}
\ee
is a well-defined function. Having the above objects at our disposal, the 
content of Proposition \ref{PROPOSITION_varphi} can be recast into a more 
convenient matrix form as follows.

\begin{LEMMA}
\label{LEMMA_F_derivative}
With the aid of the smooth functions $Z$ \eqref{Z} and $B_\mfm$ \eqref{B_mfm}, 
for the derivative of the column vector $F$ \eqref{F} along the Hamiltonian 
vector field $\bsX_H$ \eqref{bsX_H} we can write
\be
    \bsX_H [F] = (Z - B_{\mfm}) F.
\label{F_der_matrix_form}
\ee
\end{LEMMA}

\begin{proof}
Upon introducing the column vector 
\be
    J = \bsX_H [F] + B_{\mfm} F - Z F \in \bC^N,
\label{J}
\ee 
it is enough to prove that $J_k = 0$ for all $k \in \bN_N$, at each point
$(\lambda, \theta)$ of the phase space $P$ \eqref{P}. Starting with the
upper $n$ components of $J$, notice that by Proposition 
\ref{PROPOSITION_varphi} and the formulae \eqref{Z_entries}-\eqref{B_mfm} we 
can write that
\be
    J_a = \half e^{-\frac{\theta_a}{2}} u_a^{\frac{3}{2}} G_a
    \qquad
    (a \in \bN_n),
\label{J_a}
\ee
where $G_a$ is an appropriate function depending only on $\lambda$. More
precisely, it has the form
\be
\begin{split}
    G_a 
    = \Real
        \bigg(
            & \frac{2 \ri \sin(\nu)}{\sinh(2 \lambda_a) 
                \sinh(\ri \nu + 2 \lambda_a)}
            + \sum_{\substack{j = 1 \\ (j \neq a, n + a)}}^N
                \frac{\ri \sin(\mu) (1 + \bar{z}_j \bar{z}_a^{-1})}
                    {\sinh(\lambda_a - \Lambda_j) 
                    \sinh(\ri \mu + \lambda_a - \Lambda_j)}
            \\
            & + \frac{\ri \sin(\mu) (z_a \bar{z}_a^{-1} - 1)
                        + \ri \sin(\mu - \nu)(\bar{z}_a^{-1} - z_a^{-1})}
                {\sinh(2 \lambda_a) \sinh(\ri \mu + 2 \lambda_a)}
        \bigg),
\end{split}
\label{G_a}
\ee
that can be made quite explicit by exploiting the definition of the 
constituent functions $z_j$ \eqref{z_a}. Now, following the same strategy 
we applied in the proof of Proposition \ref{PROPOSITION_L_in_G}, let us 
introduce a complex valued function $g_a$ depending only on a single 
complex variable $w$, obtained simply by replacing $\lambda_a$ with 
$\lambda_a + w$ in the explicit expression of right-hand side of the 
above equation \eqref{G_a}. In mod $\ri \pi$ sense this meromorphic
function has at most first order poles at the points
\be
    w \equiv -\lambda_a, \, 
    w \equiv \pm \ri \mu / 2 - \lambda_a, \, 
    w \equiv \pm \ri \nu / 2 - \lambda_a, \,
    w \equiv \Lambda_j - \lambda_a, \, 
    w \equiv \pm (\ri \mu + \Lambda_j) - \lambda_a \, (j \in \bN_N).
\label{poles-a}
\ee
However, at each of these points the residue of $g_a$ turns out to be zero. 
Moreover, it is obvious that $g_a(w)$ vanishes as $\Real(w) \to \infty$, 
therefore Liouville's theorem implies that $g_a(w) = 0$ for all $w \in \bC$. 
In particular $G_a = g_a(0) = 0$, and so by \eqref{J_a} we conclude that 
$J_a = 0$.

Turning to the lower $n$ components of the column vector $J$ \eqref{J}, let 
us note that our previous result $J_a = 0$ allows us to write that
\be
    J_{n+a}
    = -\sinh(\theta_a) e^{-\frac{\theta_a}{2}} 
        u_a^{\half} \bar{z}_a G_{n + a}
    \qquad
    (a \in \bN_n),
\label{J_n+a}
\ee
where $G_{n + a}$ is again an appropriate smooth function depending only on 
$\lambda$, as can be seen from the formula
\begin{align}
    G_{n + a} 
    = & \frac{2 \ri \sin(\nu)}
            {\sinh(2 \lambda_a) \sinh(\ri \nu - 2 \lambda_a)}
    - \frac{\ri \sin(\mu) + \ri \sin(\mu - \nu) \bar{z}_a^{-1}}
            {\sinh(2 \lambda_a) \sinh(\ri \mu - 2 \lambda_a)}
    + \bar{z}_a^{-1} 
        \frac{\ri \sin(\mu) z_a + \ri \sin(\mu-\nu)}
            {\sinh(2 \lambda_a) \sinh(\ri \mu + 2 \lambda_a)}
    \nonumber \\
    & + \sum_{\substack{j = 1 \\ (j \neq a, n + a)}}^N
            \frac{1}{\sinh(\lambda_a - \Lambda_j)}
            \left(
                \frac{\ri \sin(\mu)}{\sinh(\ri \mu - \lambda_a + \Lambda_j)}
                + \frac{\ri \sin(\mu) \bar{z}_a^{-1} \bar{z}_j}
                    {\sinh(\ri \mu + \lambda_a - \Lambda_j)}
            \right).
\label{G_n+a}
\end{align}
Next, let us plug the definition of $z_j$ \eqref{z_a} into the above
expression \eqref{G_n+a} and introduce the complex valued function 
$g_{n + a}$ of $w \in \bC$ by replacing $\lambda_a$ with $\lambda_a + w$ 
in the resulting formula. Note that $g_{n + a}$ has at most first order 
poles at the points
\be
    w \equiv -\lambda_a, \,
    w \equiv \pm \ri \mu / 2 - \lambda_a, \, 
    w \equiv \Lambda_j - \lambda_a \, (j \in \bN_N)
    \pmod{\ri \pi},
\label{poles-n+a}
\ee
but all these singularities are removable. Since $g_{n + a}(w) \to 0$ 
as $\Real(w) \to \infty$, the boundedness of the periodic function 
$g_{n + a}$ is also obvious. Thus, Liouville's theorem entails that 
$g_{n + a}= 0$ on the whole complex plane, whence the relationship 
$G_{n + a} = g_{n + a}(0) = 0$ also follows. Now, looking back to the
equation \eqref{J_n+a}, we end up with the desired equation 
$J_{n + a} = 0$.
\end{proof}

\subsection{Lax representation of the dynamics}
\label{SUBSECTION_Lax_representation}
Based on our proposed Lax matrix \eqref{L}, in this subsection we wish to
construct a Lax representation for the dynamics of the van Diejen system 
\eqref{H}. As it turns out, Lemmas \ref{LEMMA_commut_rel} and 
\ref{LEMMA_F_derivative} prove to be instrumental in our approach. As the 
first step, by applying the Hamiltonian vector field $\bsX_H$ \eqref{bsX_H}
on the Ruijsenaars type commutation relation \eqref{commut_rel}, let us 
observe that the Leibniz rule yields
\be
\begin{split}
    & e^{\ri \mu} e^{\ad_{\bsLambda}} 
        \left( 
            \bsX_H [L] 
            - \left[ L, e^{-\bsLambda} \bsX_H [e^{\bsLambda}] \right] 
        \right) 
    - e^{-\ri \mu} e^{-\ad_{\bsLambda}} 
        \left( 
            \bsX_H [L] 
            + \left[ L, \bsX_H [e^{\bsLambda}] e^{-\bsLambda} \right] 
        \right)
    \\
    & \quad = 2 \ri \sin(\mu) 
                \left( 
                    \bsX_H [F] F^* + F (\bsX_H [F])^* 
                \right).
\end{split}
\label{commut_rel_der_1}
\ee
By comparing the formula appearing in \eqref{lambda_dot} with the matrix
entries \eqref{D_entries} of the diagonal matrix $D$, it is clear that
\be
    \bsX_H [\bsLambda] = D,
\label{bsLambda_der}
\ee 
which in turn implies that
\be
    e^{-\bsLambda} \bsX_H [e^{\bsLambda}] 
    = \bsX_H [e^{\bsLambda}] e^{-\bsLambda} 
    = D.
\label{e_bsLambda_der}
\ee
Thus, the above equation \eqref{commut_rel_der_1} can be cast into the 
fairly explicit form
\be
\begin{split}
    & e^{\ri \mu} e^{\ad_{\bsLambda}} 
        \left( 
            \bsX_H [L] - \left[ L, D \right] 
        \right)
    - e^{-\ri \mu} e^{-\ad_{\bsLambda}} 
        \left( 
            \bsX_H [L] + \left[ L, D \right] 
        \right) 
    \\
    & \quad = 2 \ri \sin(\mu) 
                \left( 
                    \bsX_H [F] F^* + F (\bsX_H [F])^* 
                \right),
\end{split}
\label{commut_rel_der_2}
\ee
which serves as the starting point in our analysis on the derivative 
$\bsX_H [L]$. Before formulating the main result of this subsection, over
the phase space $P$ \eqref{P} we define the matrix valued function
\be
    B_{\mfm^\perp} = -\coth(\wad_{\bsLambda}) Y \in \mfm^\perp.
\label{B_mfm_perp}
\ee
Recalling the definition \eqref{Y}, we see that $B_{\mfm^\perp}$ is actually 
an off-diagonal matrix. Furthermore, for its non-trivial entries we have the 
explicit expressions
\be
    (B_{\mfm^\perp})_{k, l} 
    = - \coth(\Lambda_k - \Lambda_l) \frac{L_{k, l} - (L^{-1})_{k, l}}{2}
    \qquad
    (k, l \in \bN_N, \, k \neq l).
\label{B_mfm_perp_entries}
\ee
Finally, with the aid of the diagonal matrix $B_{\mfm}$ \eqref{B_mfm}, over
the phase space $P$ \eqref{P} we also define the $\mfk$-valued smooth function
\be
    B = B_{\mfm} + B_{\mfm^\perp} \in \mfk.
\label{B}
\ee

\begin{THEOREM}
\label{THEOREM_Lax_representation}
The derivative of the matrix valued function $L$ \eqref{L} along the 
Hamiltonian vector field $\bsX_H$ \eqref{bsX_H} takes the Lax form
\be
    \bsX_H [L] = [L, B].
\label{Lax_representation}
\ee
In other words, the matrices $L$ \eqref{L} and $B$ \eqref{B} provide a Lax 
pair for the dynamics generated by the Hamiltonian \eqref{H}.
\end{THEOREM}

\begin{proof}
For simplicity, let us introduce the matrix valued smooth functions
\be
    \Psi = \bsX_H [L] - [L, B]
    \quad \text{and} \quad
    R = \sinh(\ri \mu \Id_{\mfgl(N, \bC)} + \ad_{\bsLambda}) \Psi
\label{Psi_and_R}
\ee
defined on the phase space $P$ \eqref{P}. Our goal is to prove that $\Psi = 0$. 
However, since $\sin(\mu) \neq 0$, the linear operator
\be
    \sinh(\ri \mu \Id_{\mfgl(N, \bC)} + \ad_{\bsLambda})
    \in \End(\mfgl(N, \bC))
\label{Lax_repr_linear_op}
\ee
is invertible at each point of $P$, whence it is enough to show that $R = 0$. 
For this reason, notice that from the relationship \eqref{commut_rel_der_2}
we can infer that
\be
\begin{split}
    2 R 
    = & 
        e^{\ri \mu} e^{\ad_{\bsLambda}} \Psi
        - e^{-\ri \mu} e^{-\ad_{\bsLambda}} \Psi
    \\
    = & 2 \ri \sin(\mu) \left( \bsX_H [F] F^* + F (\bsX_H [F])^* \right) 
        -\left( 
            e^{\ri \mu} e^{\ad_{\bsLambda}} [L, B_{\mfm^\perp}]
            - e^{-\ri \mu} e^{-\ad_{\bsLambda}} [L, B_{\mfm^\perp}] 
        \right)
    \\ 
    &   -\left( 
            e^{\ri \mu} e^{\ad_{\bsLambda}} [L, B_{\mfm}]
            - e^{-\ri \mu} e^{-\ad_{\bsLambda}} [L, B_{\mfm}] 
        \right)
        -\left( 
            e^{\ri \mu} e^{\ad_{\bsLambda}} [D, L]
            + e^{-\ri \mu} e^{-\ad_{\bsLambda}} [D, L] 
        \right).
\end{split}
\label{R_expanded}
\ee
Our strategy is to inspect the right-hand side of the above equation 
term-by-term.

As a preparatory step, from the definitions of $D$ \eqref{D} and
$Y$ \eqref{Y} we see that
\be
    (L - L^{-1}) / 2 = D + Y,
\label{D+Y}
\ee
thus the commutation relation
\be
    [L, Y] = [ L, -D + (L - L^{-1}) / 2 ] = [D, L]
\label{L_Y_D_comm_rel}
\ee
readily follows. Keeping in mind the relationship \eqref{L_Y_D_comm_rel}
and the standard hyperbolic functional equations
\be
    \coth(w) \pm 1 = \frac{e^{\pm w}}{\sinh(w)}
    \qquad
    (w \in \bC),
\label{coth_iden}
\ee
from the definitions of $Z$ \eqref{Z} and $B_{\mfm^\perp}$ \eqref{B_mfm_perp}
we infer that
\be
\begin{split}
    e^{\ad_{\bsLambda}} [L, B_{\mfm^\perp}]
    = & -e^{\ad_{\bsLambda}} [L, \coth(\wad_{\bsLambda}) Y]
    \\
    = & -e^{\ad_{\bsLambda}}
        \left(
            [L, (\coth(\wad_{\bsLambda}) 
                - \Id_{\mfm^\perp \oplus \mfa^\perp}) Y]
            + [L, Y]
        \right)
    \\
    = & -e^{\ad_{\bsLambda}}
        \left(
            [L, e^{-\wad_{\bsLambda}} \sinh(\wad_{\bsLambda})^{-1} Y]
            + [D, L]
        \right)
    \\
    = & -[e^{\ad_{\bsLambda}} L, Z] - e^{\ad_{\bsLambda}} [D, L].
\end{split}
\label{term_1_first_half}
\ee
Along the same lines, one finds immediately that
\be
    e^{-\ad_{\bsLambda}} [L, B_{\mfm^\perp}] 
    = -[e^{-\ad_{\bsLambda}} L, Z] + e^{-\ad_{\bsLambda}} [D, L].
\label{term_1_second_half}
\ee
At this point let us recall that $Z$ \eqref{Z} takes values in the subspace 
$\mfm^\perp \subseteq \mfk$, thus it is anti-Hermitian and commutes with the 
matrix $C$ \eqref{C}. Therefore, by utilizing equations
\eqref{term_1_first_half} and \eqref{term_1_second_half}, the application of 
the commutation relation \eqref{commut_rel} leads to the relationship
\be
\begin{split}
    & e^{\ri \mu} e^{\ad_{\bsLambda}} [L, B_{\mfm^\perp}]
    - e^{-\ri \mu} e^{-\ad_{\bsLambda}} [L, B_{\mfm^\perp}]
    \\
    & \quad 
    = -[e^{\ri \mu} e^{\ad_{\bsLambda}} L 
        - e^{-\ri \mu} e^{-\ad_{\bsLambda}} L, Z]
        - \left( 
            e^{\ri \mu} e^{\ad_{\bsLambda}} [D, L]
            + e^{-\ri \mu} e^{-\ad_{\bsLambda}} [D, L] 
        \right)
    \\
    & \quad
    = -[2 \ri \sin(\mu) F F^* + 2 \ri \sin(\mu - \nu) C, Z]
        - \left( 
            e^{\ri \mu} e^{\ad_{\bsLambda}} [D, L]
            + e^{-\ri \mu} e^{-\ad_{\bsLambda}} [D, L] 
        \right)
    \\
    & \quad
    = 2 \ri \sin(\mu) \left( (Z F) F^* + F (Z F)^* \right)
        - \left( 
            e^{\ri \mu} e^{\ad_{\bsLambda}} [D, L]
            + e^{-\ri \mu} e^{-\ad_{\bsLambda}} [D, L] 
        \right).
\end{split}
\label{term_1}
\ee
To proceed further, let us recall that $B_\mfm$ \eqref{B_mfm} takes values 
in $\mfm \subseteq \mfk$, whence it is also anti-Hermitian and also commutes 
with the matrix $C$ \eqref{C}. Thus, by applying commutation relation
\eqref{commut_rel} again, we obtain at once that
\be
\begin{split}
    & e^{\ri \mu} e^{\ad_{\bsLambda}} [L, B_{\mfm}]
    - e^{-\ri \mu} e^{-\ad_{\bsLambda}} [L, B_{\mfm}]
    \\
    & \quad 
    = e^{\ri \mu} [e^{\ad_{\bsLambda}} L, e^{\ad_{\bsLambda}} B_\mfm]
        -e^{-\ri \mu} [e^{-\ad_{\bsLambda}} L, e^{-\ad_{\bsLambda}} B_\mfm]
    = [e^{\ri \mu} e^{\ad_{\bsLambda}} L 
        - e^{-\ri \mu} e^{-\ad_{\bsLambda}} L, B_\mfm]
    \\
    & \quad
    = [2 \ri \sin(\mu) F F^* + 2 \ri \sin(\mu - \nu) C, B_\mfm]
    = -2 \ri \sin(\mu) \left( (B_\mfm F) F^* + F (B_\mfm F)^* \right).
\end{split}
\label{term_2}
\ee
Now, by plugging the expressions \eqref{term_1} and \eqref{term_2}
into \eqref{R_expanded}, we obtain that
\be
    R = \ri \sin(\mu) 
        \left(
            (\bsX_H [F] - Z F + B_\mfm F) F^*
            + F (\bsX_H [F] - Z F + B_\mfm F)^*
        \right).
\label{R_OK}
\ee
Giving a glance at Lemma \ref{LEMMA_F_derivative}, we conclude that $R = 0$, 
thus the Theorem follows.
\end{proof}

At this point we wish to make a short comment on matrix 
$B = B(\lambda, \theta; \mu, \nu)$ \eqref{B} appearing in the Lax 
representation \eqref{Lax_representation} of the dynamics \eqref{H}. It is 
an important fact that by taking its `rational limit' we can recover the 
second member of the Lax pair of the rational $C_n$ van Diejen system with 
two parameters $\mu$ and $\nu$. More precisely, up to some irrelevant 
numerical factors, in the $\alpha \to 0^+$ limit the matrix 
$\alpha B(\alpha \lambda, \theta; \alpha \mu, \alpha \nu)$ tends to the 
second member $\hat{\cB}(\lambda, \theta; \mu, \nu, \kappa = 0)$ of the 
rational Lax pair, that first appeared in equation (4.60) of the recent paper 
\cite{Pusztai_NPB2015}. In other words, matrix $B$ \eqref{B} is an appropriate 
hyperbolic generalization of the `rational' matrix $\hat{\cB}$ with two 
coupling parameters. We can safely state that the results presented in 
\cite{Pusztai_NPB2015} has played a decisive role in our present work. As a 
matter of fact, most probably we could not have guessed the form of the 
non-trivial building blocks \eqref{B_mfm} and \eqref{B_mfm_perp} without the 
knowledge of rational analogue of $B$. 

In order to harvest some consequences of the Lax representation 
\eqref{Lax_representation}, we continue with a simple corollary of Theorem 
\ref{THEOREM_Lax_representation}, that proves to be quite handy in the 
developments of the next subsection.

\begin{PROPOSITION}
\label{PROPOSITION_D_and_Y_derivatives}
For the derivatives of the matrix valued smooth functions $D$ \eqref{D}
and $Y$ \eqref{Y} along the Hamiltonian vector field $\bsX_H$ \eqref{bsX_H}
we have
\be
    \bsX_H [D] = [Y, B_{\mfm^\perp}]_{\mfa}
    \quad \text{and} \quad
    \bsX_H [Y] 
    = [Y, B_{\mfm^\perp}]_{\mfa^\perp} 
        + [D, B_{\mfm^\perp}] 
        + [Y, B_{\mfm}].
\label{D_and_Y_der}
\ee
\end{PROPOSITION}

\begin{proof}
As a consequence of Proposition \ref{PROPOSITION_L_in_G}, for the inverse of 
$L$ we can write that $L^{-1} = C L C$. Since the matrix valued function $B$ 
\eqref{B} takes values in $\mfk$ \eqref{mfk}, from Theorem 
\ref{THEOREM_Lax_representation} we infer that
\be
    \bsX [L^{-1}] 
    = C \bsX_H [L] C = C [L, B] C = [C L C, C B C] = [L^{-1}, B],
\label{L_inv_der}
\ee
thus the equation
\be
    \bsX_H [(L - L^{-1}) / 2] = [(L - L^{-1}) / 2, B]
\label{bsX_H_on_L-L_inv}
\ee
is immediate. Due to the relationship \eqref{D+Y}, by simply projecting of 
the above equation onto the subspaces $\mfa$ and $\mfa^\perp$, respectively, 
the derivatives displayed in \eqref{D_and_Y_der} follow at once.
\end{proof}

\subsection{Geodesic interpretation}
\label{SUBSECTION_Geodesic_interpretation}
The geometric study of the CMS type integrable systems goes back to the 
fundamental works of Olshanetsky and Perelomov (see e.g. 
\cite{Olsha_Pere_1976, OlshaPere}). Since their landmark papers the 
so-called projection method has been vastly generalized to cover many 
variants of the CMS type particle systems. By now some result are
available in the context of the RSvD models, too. For details, see
e.g. \cite{KKS, Feher_Pusztai_NPB2006, Feher_Pusztai_2007, 
Feher_Klimcik_0901, Pusztai_NPB2011, Pusztai_NPB2012}. The primary goal 
of this subsection is to show that the Hamiltonian flow generated by the 
Hamiltonian \eqref{H} can be also obtained by an appropriate `projection 
method' from the geodesic flow of the Lie group $U(n ,n)$. In order to 
make this statement more precise, take the maximal integral curve
\be
    \bR \ni 
        t 
        \mapsto (
        \lambda(t), \theta(t)) 
        = (\lambda_1(t), \ldots, \lambda_n(t), 
            \theta_1(t), \ldots, \theta_n(t))
    \in P
\label{trajectory}
\ee
of the Hamiltonian vector field $\bsX_H$ \eqref{bsX_H} satisfying the initial 
condition
\be
    \gamma(0) = \gamma_0,
\label{trajectory_init_cond}
\ee
where $\gamma_0 \in P$ is an arbitrary point. By exploiting Proposition 
\ref{PROPOSITION_D_and_Y_derivatives}, we start our analysis with the 
following observation.

\begin{PROPOSITION}
\label{PROPOSITION_ddot_bsLambda_etc}
Along the maximally defined trajectory \eqref{trajectory}, the time evolution
of the diagonal matrix $\bsLambda = \bsLambda(t) \in \mfc$ \eqref{bsLambda} 
obeys the second order differential equation
\be
    \ddot{\bsLambda} + [Y, \coth(\wad_{\bsLambda}) Y]_\mfa = 0,
\label{bsLambda_DE}
\ee
whilst for the evolution of $Y = Y(t)$ \eqref{Y} we have the first order 
equation
\be
    \dot{Y} + [Y, \coth(\wad_{\bsLambda}) Y]_{\mfa^\perp} - [Y, B_{\mfm}]
        + [\dot{\bsLambda}, \coth(\wad_{\bsLambda}) Y] = 0.
\label{Y_DE}
\ee
\end{PROPOSITION}

\begin{proof}
Due to equation \eqref{bsLambda_der}, along the solution curve 
\eqref{trajectory} we can write 
\be
    \dot{\bsLambda} = D, 
\label{bsLambda_dot} 
\ee
whereas from the relationships displayed in \eqref{D_and_Y_der} we get
\be
    \dot{D} = [Y, B_{\mfm^\perp}]_{\mfa}
    \quad \text{and} \quad
    \dot{Y} 
    = [Y, B_{\mfm^\perp}]_{\mfa^\perp} 
        + [D, B_{\mfm^\perp}] 
        + [Y, B_{\mfm}].
\label{D_dot_and_Y_dot}
\ee
Recalling the definition \eqref{B_mfm_perp}, equations \eqref{bsLambda_DE}
and \eqref{Y_DE} clearly follow.
\end{proof}

Next, by evaluating the matrices $Z$ \eqref{Z} and $B_\mfm$ \eqref{B_mfm}
along the fixed trajectory \eqref{trajectory}, for all $t \in \bR$ we define
\be
    \cK(t) = B_{\mfm}(t) - Z(t) \in \mfk.
\label{cK}
\ee
Since the dependence of $\cK$ on $t$ is smooth, there is a unique maximal
smooth solution 
\be
    \bR \ni t \mapsto k(t) \in GL(N, \bC)
\label{k}
\ee
of the first order differential equation
\be
    \dot{k}(t) = k(t) \cK(t)
    \qquad
    (t \in \bR)
\label{k_DE}
\ee
satisfying the initial condition 
\be
    k(0) = \bsone_N. 
\label{k_initial_cond}
\ee
Since \eqref{k_DE} is a linear differential equation for $k$, the existence 
of such a global fundamental solution is obvious. Moreover, since $\cK$ 
\eqref{cK} takes values in the Lie algebra $\mfk$ \eqref{mfk}, the trivial 
observations
\be
    \frac{\dd (k C k^*)}{\dd t}
    = \dot{k} C k^* + k C \dot{k}^*
    = k (\cK C + C \cK^*) k^*
    = 0
    \quad \text{and} \quad
    k(0) C k(0)^* = C
\label{k_quadratic_eq}
\ee
imply immediately that $k$ \eqref{k} actually takes values in the subgroup 
$K$ \eqref{K}; that is,
\be
    k(t) \in K
    \qquad
    (t \in \bR).
\label{k_in_K}
\ee
Utilizing $k$, we can formulate the most important technical result of this 
subsection.

\begin{LEMMA}
\label{LEMMA_A}
The smooth function
\be
    \bR \ni t 
        \mapsto 
        A(t) = k(t) e^{2 \bsLambda(t)} k(t)^{-1}
    \in \exp(\mfp_\reg)
\label{A}
\ee
satisfies the second order geodesic differential equation
\be
    \frac{\dd}{\dd t} \left( \frac{\dd A(t)}{\dd t} A(t)^{-1} \right) = 0
    \qquad
    (t \in \bR).
\label{A_geodesic_eqn}
\ee
\end{LEMMA}

\begin{proof}
First, let us observe that \eqref{A} is a well-defined map. Indeed, since 
along the trajectory \eqref{trajectory} we have $\bsLambda(t) \in \mfc$, 
from \eqref{mfp_reg_identification} we see that $A$ does take values in 
$\exp(\mfp_\reg)$. Continuing with the proof proper, notice that for all 
$t \in \bR$ we have $A^{-1} = k e^{-2 \bsLambda} k^{-1}$ and
\be
    \dot{A} 
    = \dot{k} e^{2 \bsLambda} k^{-1}
        + k e^{2 \bsLambda} 2 \dot{\bsLambda} k^{-1}
        - k e^{2 \bsLambda} k^{-1} \dot{k} k^{-1},
\label{A_dot}
\ee
thus the formulae
\be
    \dot{A} A^{-1}  
    = k \left( 
            2 \dot{\bsLambda} - e^{2 \ad_{\bsLambda}} \cK + \cK 
        \right) k^{-1}
    \quad \text{and} \quad
    A^{-1} \dot{A} 
    = k \left( 
            2 \dot{\bsLambda} + e^{-2 \ad_{\bsLambda}} \cK - \cK 
        \right) k^{-1}
\label{A_inv_and_A_dot}
\ee
are immediate. Upon introducing the shorthand notations
\begin{align}
    & \cL(t) = \dot{\bsLambda}(t) + \cosh(\wad_{\bsLambda(t)}) Y(t) \in \mfp,
    \label{cL} \\
    & \cN(t) = \sinh(\wad_{\bsLambda(t)}) Y(t) \in \mfk,
    \label{cN}
\end{align}
from \eqref{A_inv_and_A_dot} we conclude that
\be
\begin{split}
    \frac{\dot{A} A^{-1} + A^{-1} \dot{A}}{4} 
    & = k \left( 
            \dot{\bsLambda} - \half \sinh(2 \ad_{\bsLambda}) \cK 
            \right) k^{-1} 
    = k \left( 
            \dot{\bsLambda} - \half \sinh(2 \wad_{\bsLambda}) \cK_{\mfm^\perp} 
            \right) k^{-1}
    \\
    & = k \left( 
            \dot{\bsLambda} 
            + \cosh(\wad_{\bsLambda}) \sinh(\wad_{\bsLambda}) Z
            \right) k^{-1}
    = k \cL k^{-1},
\end{split}
\label{A_dot_+}
\ee
and the relationship
\be
\begin{split}
    \frac{\dot{A} A^{-1} - A^{-1} \dot{A}}{4} 
    & = k \frac{\cK - \cosh(2 \ad_{\bsLambda}) \cK}{2} k^{-1} 
    = - k \left( \sinh(\ad_{\bsLambda})^2 \cK \right) k^{-1} 
    \\
    & = k \left( \sinh(\wad_{\bsLambda})^2 Z \right) k^{-1}
    = k \cN k^{-1}
\end{split}
\label{A_dot_-}
\ee
also follows. 

Now, by differentiating \eqref{A_dot_+} with respect to time $t$, we get
\be
    \frac{\dd}{\dd t} \frac{\dot{A} A^{-1} + A^{-1} \dot{A}}{4}
    = k \left (\dot{\cL} - [\cL, \cK] \right) k^{-1}.
\label{A_ddot_+}
\ee
Recalling the definition \eqref{cL}, Leibniz rule yields
\be
    \dot{\cL}
    = \ddot{\bsLambda} 
        + [\dot{\bsLambda}, \sinh(\wad_{\bsLambda}) Y] 
        + \cosh(\wad_{\bsLambda}) \dot{Y},
\label{cL_dot}
\ee
and the commutator
\be
    [\cL, \cK] 
    = - [\dot{\bsLambda}, \sinh(\wad_{\bsLambda})^{-1} Y]
        + [\cosh(\wad_{\bsLambda}) Y, B_{\mfm}]
        - [\cosh(\wad_{\bsLambda}) Y, \sinh(\wad_{\bsLambda})^{-1} Y]
\label{cL_and_cK_commutator}
\ee
is also immediate. By inspecting the right-hand side of the above equation, 
for the second term one can easily derive that
\be
\begin{split}
    [\cosh(\wad_{\bsLambda}) Y, B_{\mfm}] 
    & = \half [e^{\ad_{\bsLambda}} Y, B_{\mfm}] 
        + \half [e^{-\ad_{\bsLambda}} Y, B_{\mfm}]
    = \half e^{\ad_{\bsLambda}} [Y, B_{\mfm}]
        + \half e^{-\ad_{\bsLambda}} [Y, B_{\mfm}]
    \\
    & = \cosh(\ad_{\bsLambda}) [Y, B_{\mfm}]
    = \cosh(\wad_{\bsLambda}) [Y, B_{\mfm}].
\end{split}
\label{comm_term_1}
\ee
Furthermore, bearing in mind the identities appearing in \eqref{coth_iden}, 
a slightly longer calculation also reveals that the third term in 
\eqref{cL_and_cK_commutator} can be cast into the form
\be
\begin{split}
    & [\cosh(\wad_{\bsLambda}) Y, \sinh(\wad_{\bsLambda})^{-1} Y]
    = \half [e^{\ad_{\bsLambda}} Y, \sinh(\wad_{\bsLambda})^{-1} Y]
        + \half [e^{-\ad_{\bsLambda}} Y, \sinh(\wad_{\bsLambda})^{-1} Y]
    \\
    & \quad
    = \half e^{\ad_{\bsLambda}} 
                [Y, e^{-\wad_{\bsLambda}} \sinh(\wad_{\bsLambda})^{-1} Y]
        + \half e^{-\ad_{\bsLambda}} 
                [Y, e^{\wad_{\bsLambda}} \sinh(\wad_{\bsLambda})^{-1} Y]
    \\
    & \quad
    = \cosh(\ad_{\bsLambda}) [Y, \coth(\wad_{\bsLambda}) Y]
    = [Y, \coth(\wad_{\bsLambda}) Y]_{\mfa} 
        + \cosh(\wad_{\bsLambda}) [Y, \coth(\wad_{\bsLambda}) Y]_{\mfa^\perp}. 
\end{split}
\label{comm_term_2}
\ee
Now, by plugging the expressions \eqref{comm_term_1} and \eqref{comm_term_2}
into \eqref{cL_and_cK_commutator}, and by applying the hyperbolic identity
\be
    \sinh(w) + \frac{1}{\sinh(w)} = \cosh(w) \coth(w)
    \qquad
    (w \in \bC),
\label{sinh_iden}
\ee
one finds immediately that
\be
\begin{split}
    \dot{\cL} - [\cL, \cK]
    = & \ddot{\bsLambda} + [Y, \coth(\wad_{\bsLambda}) Y]_{\mfa}
    \\ 
    & + \cosh(\wad_{\bsLambda}) 
        \left( 
            \dot{Y} 
            + [Y, \coth(\wad_{\bsLambda}) Y]_{\mfa^\perp}
            - [Y, B_{\mfm}]
            + [\dot{\bsLambda}, \coth(\wad_{\bsLambda}) Y]
        \right).
\end{split}
\label{cL_key}
\ee
Looking back to Proposition \ref{PROPOSITION_ddot_bsLambda_etc}, we see
that $\dot{\cL} - [\cL, \cK] = 0$, thus by \eqref{A_ddot_+} we end up with 
the equation
\be
    \frac{\dd}{\dd t} \frac{\dot{A} A^{-1} + A^{-1} \dot{A}}{4} = 0.
\label{A_ddot_+_OK}
\ee

Next, upon differentiating \eqref{A_dot_-} with respect to $t$, we see that
\be
    \frac{\dd}{\dd t} \frac{\dot{A} A^{-1} - A^{-1} \dot{A}}{4}
    = k \left (\dot{\cN} - [\cN, \cK] \right) k^{-1}.
\label{A_ddot_-}
\ee
Remembering the form of $\cN$ \eqref{cN}, Leibniz rule yields
\be
    \dot{\cN} 
    = \cosh(\wad_{\bsLambda}) [\dot{\bsLambda}, Y] 
        + \sinh(\wad_{\bsLambda}) \dot{Y}
    = \sinh(\wad_{\bsLambda})
        \left(
            \coth(\wad_{\bsLambda}) [\dot{\bsLambda}, Y] + \dot{Y}
        \right),
\label{cN_dot}
\ee
and the formula
\be
    [\cN, \cK] 
    = [\sinh(\wad_{\bsLambda}) Y, B_{\mfm}] 
        - [\sinh(\wad_{\bsLambda}) Y, \sinh(\wad_{\bsLambda})^{-1} Y]
\label{cN_and_cK_commutator}
\ee
is also immediate. Now, let us observe that the first term on the 
right-hand side of the above equation can be transformed into the 
equivalent form
\be
\begin{split}
    [\sinh(\wad_{\bsLambda}) Y, B_{\mfm}] 
    & = \half [e^{\ad_{\bsLambda}} Y, B_{\mfm}] 
        - \half [e^{-\ad_{\bsLambda}} Y, B_{\mfm}]
    = \half e^{\ad_{\bsLambda}} [Y, B_{\mfm}]
        - \half e^{-\ad_{\bsLambda}} [Y, B_{\mfm}]
    \\
    & = \sinh(\ad_{\bsLambda}) [Y, B_{\mfm}]
    = \sinh(\wad_{\bsLambda}) [Y, B_{\mfm}],
\end{split}
\label{comm_term_3}
\ee
while for the second term we get
\be
\begin{split}
    & [\sinh(\wad_{\bsLambda}) Y, \sinh(\wad_{\bsLambda})^{-1} Y]
    = \half [e^{\ad_{\bsLambda}} Y, \sinh(\wad_{\bsLambda})^{-1} Y]
        - \half [e^{-\ad_{\bsLambda}} Y, \sinh(\wad_{\bsLambda})^{-1} Y]
    \\
    & \quad
    = \half e^{\ad_{\bsLambda}} 
                [Y, e^{-\wad_{\bsLambda}} \sinh(\wad_{\bsLambda})^{-1} Y]
        - \half e^{-\ad_{\bsLambda}} 
                [Y, e^{\wad_{\bsLambda}} \sinh(\wad_{\bsLambda})^{-1} Y]
    \\
    & \quad
    = \sinh(\ad_{\bsLambda}) [Y, \coth(\wad_{\bsLambda}) Y]
    = \sinh(\wad_{\bsLambda}) [Y, \coth(\wad_{\bsLambda}) Y]_{\mfa^\perp}. 
\end{split}
\label{comm_term_4}
\ee
Taking into account the above expressions, we obtain that
\be
\begin{split}
    \dot{\cN} - [\cN, \cK]
    = \sinh(\wad_{\bsLambda}) 
        \left( 
            \dot{Y} 
            + [Y, \coth(\wad_{\bsLambda}) Y]_{\mfa^\perp}
            - [Y, B_{\mfm}]
            + [\dot{\bsLambda}, \coth(\wad_{\bsLambda}) Y]
        \right),
\end{split}
\label{cN_key}
\ee
whence by Proposition \ref{PROPOSITION_ddot_bsLambda_etc} we are entitled 
to write that $\dot{\cN} - [\cN, \cK] = 0$. Giving a glance at the 
relationship \eqref{A_ddot_-}, it readily follows that
\be
    \frac{\dd}{\dd t} \frac{\dot{A} A^{-1} - A^{-1} \dot{A}}{4} = 0.
\label{A_ddot_-_OK}
\ee
To complete the proof, observe that the desired geodesic equation 
\eqref{A_geodesic_eqn} is a trivial consequence of the equations
\eqref{A_ddot_+_OK} and \eqref{A_ddot_-_OK}.
\end{proof}

To proceed further, let us observe that by integrating the differential
equation \eqref{A_geodesic_eqn}, we obtain immediately that
\be
    \dot{A}(t) A(t)^{-1} = \dot{A}(0) A(0)^{-1}
    \qquad
    (t \in \bR).
\label{A_first_order_DE}
\ee
However, recalling the definitions \eqref{cL} and \eqref{cN}, and also the 
relationships \eqref{bsLambda_dot} and \eqref{D+Y}, from the equations 
\eqref{A_dot_+}, \eqref{A_dot_-} and \eqref{k_initial_cond} we infer that
\be
\begin{split}
    \dot{A}(0) A(0)^{-1} 
    & = 2 k(0) (\cL(0) + \cN(0)) k(0)^{-1}
    = 2 (\dot{\bsLambda}(0) + e^{\ad_{\bsLambda(0)}} Y(0))
    \\
    & = 2 e^{\ad_{\bsLambda(0)}} (D(0) + Y(0))
    = e^{\bsLambda(0)} (L(0) - L(0)^{-1}) e^{-\bsLambda(0)}.
\end{split}
\label{A_dot_t=0}
\ee
Moreover, remembering \eqref{k_initial_cond} and the definition \eqref{A}, 
at $t = 0$ we can also write that
\be
    A(0) = k(0) e^{2 \bsLambda(0)} k(0)^{-1} = e^{2 \bsLambda(0)}.
\label{A(0)}
\ee
Putting the above observations together, it is now evident that the unique 
maximal solution of the first order differential equation 
\eqref{A_first_order_DE} with the initial condition \eqref{A(0)} is the 
smooth curve
\be
    A(t) 
    = e^{t e^{\bsLambda(0)} (L(0) - L(0)^{-1}) e^{-\bsLambda(0)}} 
        e^{2 \bsLambda(0)}
    = e^{\bsLambda(0)} e^{t (L(0) - L(0)^{-1})} e^{\bsLambda(0)}
    \qquad
    (t \in \bR).
\label{A(t)_exp_form}
\ee
Comparing this formula with \eqref{A}, the following result is immediate.

\begin{THEOREM}
\label{THEOREM_eigenvalue_dynamics}
Take an arbitrary maximal solution \eqref{trajectory} of the van Diejen 
system \eqref{H}, then at each $t \in \bR$ it can be recovered uniquely  
from the spectral identification
\be
    \{ e^{\pm 2 \lambda_a(t)} \, | \, a \in \bN_n \}
    = \mathrm{Spec} 
        (e^{\bsLambda(0)} e^{t (L(0) - L(0)^{-1})} e^{\bsLambda(0)}).
\label{spectral_identification}
\ee
\end{THEOREM}

The essence of the above theorem is that any solution \eqref{trajectory} of 
the van Diejen system \eqref{H} can be obtained by a purely algebraic process 
based on the diagonalization of a matrix flow. Indeed, once one finds the 
evolution of $\lambda(t)$ from \eqref{spectral_identification}, the evolution 
of $\theta(t)$ also becomes accessible by the formula
\be
    \theta_a(t) 
    = \mathrm{arcsinh} 
        \left( \frac{\dot{\lambda}_a(t)}{u_a(\lambda(t))} \right)
    \qquad
    (a \in \bN_n), 
\label{theta_evolution}
\ee
as dictated by the equation of motion \eqref{lambda_dot}.

\subsection{Temporal asymptotics}
\label{SUBSECTION_asymptotics}
One of the immediate consequences of the projection method formulated
in the previous subsection is that the Hamiltonian \eqref{H} describes 
a `repelling' particle system, thus it is fully justified to inquire
about its scattering properties. Although rigorous scattering theory is 
in general a hard subject, a careful study of the algebraic solution 
algorithm described in Theorem \ref{THEOREM_eigenvalue_dynamics} allows 
us to investigate the asymptotic properties of any maximally defined 
trajectory \eqref{trajectory} as $t \to \pm \infty$. In this respect our 
main tool is Ruijsenaars' theorem on the spectral asymptotics of exponential 
type matrix flows (see Theorem A2 in \cite{Ruij_CMP1988}). To make it work, 
let us look at the relationship \eqref{L_diagonalized} and Lemma 
\ref{LEMMA_regularity}, from where we see that there is a group element 
$y \in K$ and a unique real $n$-tuple 
$\htheta = (\htheta_1, \ldots, \htheta_n) \in \bR^n$ satisfying 
\be
    \htheta_1 > \ldots > \htheta_n > 0, 
\label{theta_order}
\ee
such that with the (regular) diagonal matrix $\hbsTheta \in \mfc$ defined 
in \eqref{hbsTheta} we can write that
\be
    L(0) = y e^{2 \hbsTheta} y^{-1}.
\label{L(0)_diagonalized}
\ee
Following the notations of the previous subsection, here $L(0)$ still stands 
for the Lax matrix \eqref{L} evaluated along the trajectory \eqref{trajectory} 
at $t = 0$. Since
\be
    L(0) - L(0)^{-1} = 2 y \sinh(2 \hbsTheta) y^{-1},
\label{exponent}
\ee
with the aid of the positive definite matrix
\be
    \hat{L} = y^{-1} e^{2 \bsLambda(0)} y \in \exp(\mfp)
\label{hat_L_0}
\ee
for the spectrum of the matrix flow appearing in \eqref{A(t)_exp_form} 
we obtain at once that
\be
    \mathrm{Spec} 
        (e^{\bsLambda(0)} e^{t (L(0) - L(0)^{-1})} e^{\bsLambda(0)})
    = \mathrm{Spec} (\hat{L} e^{2 t \sinh(2 \hbsTheta)}).
\label{spectrum_of_matrix_flow}
\ee
In order to make a closer contact with Ruijsenaars' theorem, let us also
introduce the Hermitian $n \times n$ matrix $\cR$ with entries
\be
    \cR_{a, b} = \delta_{a + b, n + 1}.
\label{cR}
\ee
Since $\cR^2 = \bsone_n$, we have $\cR^{-1} = \cR$, whence the 
block-diagonal matrix
\be
    \cW = \begin{bmatrix}
        \bsone_n & 0_n \\
        0_n & \cR_n
    \end{bmatrix} 
    \in GL(N, \bC),
\label{cW}
\ee
also satisfies the relations $\cW^{-1} = \cW = \cW^*$. As the most important 
ingredients of our present analysis, now we introduce the matrices
\be
    \bsTheta^+ = 2 \cW \hbsTheta \cW^{-1}
    \quad \text{and} \quad
    \tilde{L} = \cW \hat{L} \cW^{-1}.
\label{bsTheta_+_and_L_tilde}
\ee
Recalling the relationships \eqref{spectral_identification} and 
\eqref{spectrum_of_matrix_flow}, it is clear that for all $t \in \bR$
we can write that
\be
    \{ e^{\pm 2 \lambda_a(t)} \, | \, a \in \bN_n \}
    = \mathrm{Spec} 
        (
            \tilde{L} e^{2 t \sinh(\bsTheta^+)}
        ).
\label{spectral_identification_OK}
\ee
However, upon performing the conjugations with the unitary matrix $\cW$
\eqref{cW} in the defining equations displayed in 
\eqref{bsTheta_+_and_L_tilde}, we find immediately that
\be
    \bsTheta^+ 
    = \diag(\theta_1^+, \ldots, \theta_n^+, 
            - \theta_n^+, \ldots, -\theta_1^+),
\label{bsTheta_+}
\ee
where
\be
    \theta_a^+ = 2 \htheta_a
    \qquad
    (a \in \bN_n).
\label{theta_+}
\ee
The point is that, due to our regularity result formulated in Lemma
\ref{LEMMA_regularity}, the diagonal matrix \eqref{bsTheta_+} has a simple 
spectrum, and its eigenvalues are in strictly decreasing order along the 
diagonal (see \eqref{theta_order}). Moreover, since $\hat{L}$ 
\eqref{hat_L_0} is positive definite, so is $\tilde{L}$. In particular, the 
leading principal minors of matrix $\tilde{L}$ are all strictly positive. 
So, the exponential type matrix flow
\be
    \bR \ni t \mapsto \tilde{L} e^{2 t \sinh(\bsTheta^+)} \in GL(N, \bC)
\label{matrix_flow}
\ee
does meet all the requirements of Ruijsenaars' aforementioned theorem.
Therefore, essentially by taking the logarithm of the quotients of 
the consecutive leading principal minors of the $n \times n$ submatrix taken 
from the upper-left-hand corner of $\tilde{L}$, one finds a unique real 
$n$-tuple
\be
    \lambda^+ = (\lambda_1^+, \dots, \lambda_n^+) \in \bR^n
\label{lambda_+}
\ee
such that for all $a \in \bN_n$ we can write 
\be
    \lambda_a(t) \sim t \sinh(\theta_a^+) + \lambda_a^+
    \quad \text{and} \quad
    \theta_a(t) \sim \theta_a^+,
\label{lambda_asymptotics_t_to_infty}
\ee
up to exponentially vanishing small terms as $t \to \infty$. It is obvious 
that the same ideas work for the case $t \to -\infty$, too, with complete
control over the asymptotic momenta $\theta_a^-$ and the asymptotic phases 
$\lambda_a^-$ as well. The above observations can be summarized as follows.

\begin{LEMMA}
\label{LEMMA_asymptotics}
For an arbitrary maximal solution \eqref{trajectory} of the hyperbolic 
$n$-particle van Diejen system \eqref{H} the particles move asymptotically 
freely as $\vert t \vert \to \infty$. More precisely, for all $a \in \bN_n$ 
we have the asymptotics
\be
    \lambda_a(t) \sim t \sinh(\theta_a^\pm) + \lambda_a^\pm
    \quad \text{and} \quad
    \theta_a(t) \sim \theta_a^\pm
    \qquad
    (t \to \pm \infty),
\label{asymptotic_result}
\ee
where the asymptotic momenta obey
\be
    \theta_a^- = -\theta_a^+
    \quad \text{and} \quad
    \theta_1^+ > \ldots > \theta_n^+ > 0.
\label{asymptotic_momenta_conds}
\ee
\end{LEMMA}

We find it quite remarkable that, up to an overall sign, the asymptotic 
momenta are preserved \eqref{asymptotic_momenta_conds}. Following 
Ruijsenaars' terminology \cite{Ruij_CMP1988, Ruij_FiniteDimSolitonSystems}, 
we may say that the $2$-parameter family of van Diejen systems \eqref{H} are
\emph{finite dimensional pure soliton systems}. Now, let us remember that 
for each pure soliton system analyzed in the earlier literature, the 
scattering map has a factorized form. That is, the $n$-particle scattering 
can be completely reconstructed from the $2$-particle processes, and also 
by the $1$-particle scattering on the external potential (see e.g. 
\cite{Kulish_1976, Moser_1977, Ruij_CMP1988, Ruij_FiniteDimSolitonSystems, 
Pusztai_NPB2013}). Albeit the results we shall present in rest of the paper 
do not rely on this peculiar feature of the scattering process, still, it 
would be of considerable interest to prove this property for the hyperbolic 
van Diejen systems \eqref{H}, too. However, because of its subtleties, we 
wish to work out the details of the scattering theory in a later publication.

\section{Spectral invariants of the Lax matrix}
\label{SECTION_Spectral_invariants}
The ultimate goal of this section is to prove that the eigenvalues of the 
Lax matrix $L$ \eqref{L} are in involution. Superficially, one could say
that it follows easily from the scattering theoretical results presented
in the previous section. A convincing argument would go as follows. Recalling 
the notations \eqref{trajectory} and \eqref{trajectory_init_cond}, let us 
consider the flow
\be
    \Phi \colon \bR \times P \rightarrow P,
    \qquad
    (t, \gamma_0) \mapsto \Phi_t (\gamma_0) = \gamma(t)
\label{Phi_flow}
\ee
generated by the Hamiltonian vector field $\bsX_H$ \eqref{bsX_H}. Since 
for all $t \in \bR$ the map $\Phi_t \colon P \rightarrow P$ is a 
symplectomorphism, for all $a ,b \in \bN_n$ we can write that
\be
    \{ \theta_a \circ \Phi_t, \theta_b \circ \Phi_t \}
    = \{ \theta_a, \theta_b \} \circ \Phi_t 
    = 0.
\label{theta_and_Phi}
\ee
On the other hand, from \eqref{asymptotic_result} it is also clear that 
at each point of the phase space $P$, for all $c \in \bN_n$ we have
\be
    \theta_c \circ \Phi_t \to \theta_c^+
    \qquad
    (t \to \infty).    
\label{theta_and_theta_+}
\ee
Recalling \eqref{htheta_smooth} and \eqref{theta_+}, it is evident that 
$\theta_c^+ \in C^\infty(P)$. Therefore, by a `simple interchange of limits', 
from \eqref{theta_and_Phi} and \eqref{theta_and_theta_+} one could infer 
that the asymptotic momenta $\theta_c^+$ $(c \in \bN_n)$ Poisson commute. 
Bearing in mind the relationships \eqref{theta_+} and \eqref{spec_L}, it 
would also follow that the eigenvalues of $L$ \eqref{L} generate a maximal 
Abelian Poisson subalgebra. However, to justify the interchange of limits, 
one does need a deeper knowledge about the scattering properties than the 
pointwise limit formulated in \eqref{theta_and_theta_+}. Since we wish 
to work out the full scattering theory elsewhere, in this paper we choose an 
alternative approach by merging the temporal asymptotics of the trajectories 
with van Diejen's earlier results \cite{van_Diejen_ComposMath, 
van_Diejen_TMP1994, van_Diejen_JMP1995}.

\subsection{Link to the $5$-parameter family of van Diejen systems}
\label{SUBSECTION_link_to_vD}
As is known from the seminal papers \cite{van_Diejen_TMP1994, 
van_Diejen_JMP1995}, the definition of the classical hyperbolic van Diejen 
system is based on the smooth functions 
$v, w \colon \bR \setminus \{ 0 \} \rightarrow \bC$ defined by the formulae
\be
    v(x) = \frac{\sinh(\ri g + x)}{\sinh(x)},
    \quad
    w(x) = \frac{\sinh(\ri g_0 + x)}{\sinh(x)} 
            \frac{\cosh(\ri g_1 + x)}{\cosh(x)}
            \frac{\sinh(\ri g'_0 + x)}{\sinh(x)}
            \frac{\cosh(\ri g'_1 + x)}{\cosh(x)},
\label{v-w}
\ee
where the five independent real numbers $g$, $g_0$, $g_1$, $g_0'$, $g_1'$ 
are the coupling constants. Parameter $g$ in the `potential' function $v$ 
controls the strength of inter-particle interaction, whereas the remaining 
four constants appearing in the `external potential' $w$ are responsible 
for the influence of the ambient field. Conforming to the notations 
introduced in the aforementioned papers, let us recall that the set of 
Poisson commuting functions found by van Diejen can be succinctly written as
\be
    H_l 
    = \sum_{\substack{J \subseteq \bN_n, \ \vert J \vert \leq l 
            \\ \eps_j = \pm 1, \ j \in J}}
                \cosh(\theta_{\eps J}) 
                \vert V_{\eps J; J^c} \vert
                U_{J^c, l - \vert J \vert}
    \qquad
    (l \in \bN_n),
\label{H_vD}
\ee
where the various constituents are defined by the formulae
\be
    \theta_{\eps J} = \sum_{j \in J} \eps_j \theta_j,
    \quad
    V_{\eps J; J^c} 
    = \prod_{j \in J} w(\eps_j \lambda_j)
        \prod_{\substack{j, j' \in J \\ (j < j')}}
            v(\eps_j \lambda_j + \eps_{j'} \lambda_{j'})^2
        \prod_{\substack{j \in J \\ k \in J^c}}
            v(\eps_j \lambda_j + \lambda_k) 
            v(\eps_j \lambda_j - \lambda_k),
\label{V}
\ee
together with the expression
\be
    U_{J^c, l - \vert J \vert}
    = (-1)^{l - \vert J \vert} 
    \mkern-25mu 
    \sum_{\substack{I \subseteq J^c, \ \vert I \vert = l - \vert J \vert 
                    \\ \eps_i = \pm 1, \ i \in I}}
        \prod_{i \in I} w(\eps_i \lambda_i)
        \mkern-5mu 
        \prod_{\substack{i, i' \in I \\ (i < i')}}
            \vert v(\eps_i \lambda_i + \eps_{i'} \lambda_{i'}) \vert^2
        \mkern-5mu 
        \prod_{\substack{i \in I \\ k \in J^c \setminus I}}
        \mkern-5mu
            v(\eps_i \lambda_i + \lambda_k) 
            v(\eps_i \lambda_i-\lambda_k).
\label{U}
\ee
At this point two short technical remarks are in order. First, we extend 
the family of the first integrals \eqref{H_vD} with the constant function 
$H_0 = 1$. Analogously, in the last equation \eqref{U} it is understood 
that $U_{J^c, 0} = 1$. 

To make contact with the $2$-parameter family of van Diejen systems of 
our interest \eqref{H}, for the coupling parameters of the potential 
functions \eqref{v-w} we make the special choice 
\be
    g = \mu,
    \quad
    g_0 = g_1 = \frac{\nu}{2},
    \quad
    g'_0 = g'_1 = 0.
\label{2parameters}
\ee
Under this assumption, from the definitions \eqref{z_a} and \eqref{V} it is
evident that with the singleton $J = \{ a \}$ we can write that
\be
    V_{\{ a \}; \{ a \}^c} = -z_a
    \qquad 
    (a \in \bN_n).
\label{V-z}
\ee
Giving a glance at \eqref{U}, it is also clear that the term corresponding 
to $J = \emptyset$ in the defining sum of $H_1$ \eqref{H_vD} is a constant
function of the form
\be
    U_{\bN_n, 1} = 2 \sum_{a = 1}^n \Real(z_a)
    = - 2 \cos \left( \nu + (n - 1) \mu \right) 
        \frac{\sin(n \mu)}{\sin(\mu)}.
\label{U-z}
\ee
Plugging the above formulae into van Diejen's main Hamiltonian $H_1$ 
\eqref{H_vD}, one finds immediately that
\be
    H_1 + 2 \cos \big(\nu + (n - 1) \mu \big) \frac{\sin(n \mu)}{\sin(\mu)}
    = 2 H 
    =\tr(L).
\label{H1_vs_H}
\ee
That is, up to some irrelevant constants, our Hamiltonian $H$ \eqref{H} 
can be identified with $H_1$ \eqref{H_vD}, provided the coupling parameters 
are related by the equations displayed in \eqref{2parameters}. At this 
point one may suspect that the quantities $\tr(L^l)$ are also expressible 
with the aid of the Poisson commuting family of functions $H_l$ \eqref{H_vD}. 
Clearly, it would imply immediately that the eigenvalues of the Lax matrix 
$L$ \eqref{L} are in involution. However, due to the complexity of the 
underlying objects \eqref{V}-\eqref{U}, this naive approach would lead to 
a formidable combinatorial task, that we do not wish to pursue in this 
paper. To circumvent the difficulties, below we rather resort to a clean 
analytical approach by exploiting the scattering theoretical results 
formulated in the previous section.

\subsection{Poisson brackets of the eigenvalues of $L$}
\label{SUBSECTION_eigenvalue_PBs}
Take an arbitrary point $\gamma_0 \in P$ and consider the unique maximal 
integral curve 
\be
    \bR \ni t \mapsto \gamma(t) = (\lambda(t), \theta(t)) \in P
\label{curve_gamma}
\ee
of the Hamiltonian vector field $\bsX_H$ \eqref{bsX_H} satisfying the 
initial condition 
\be
    \gamma(0) = \gamma_0.
\label{init_cond}
\ee 
Since the functions $H_l$ \eqref{H_vD} are first integrals of the dynamics, 
their values at the point $\gamma_0$ can be recovered by inspecting the 
limit of $H_l (\gamma(t))$ as $t \to \infty$. Now, recalling the potentials 
\eqref{H_vD} and the specialization of the coupling parameters 
\eqref{2parameters}, it is evident that
\be
    \lim_{x \to \pm \infty} v(x) = e^{\pm \ri \mu}
    \quad \text{and} \quad
    \lim_{x \to \pm \infty} w(x) = e^{\pm \ri \nu}.
\label{lim_v_w}
\ee
Therefore, taking into account the regularity properties 
\eqref{asymptotic_momenta_conds} of the asymptotic momenta $\theta_c^+$
\eqref{asymptotic_result}, from Lemma \ref{LEMMA_asymptotics} and the 
definitions \eqref{H_vD}-\eqref{U} one finds immediately that
\be
    H_l(\gamma_0)
    = \lim_{t \to \infty} H_l (\gamma(t))
    = \sum_{\substack{J \subseteq \bN_n, \ \vert J \vert \leq l 
                        \\ \eps_j = \pm 1, \ j \in J}}
        \cosh(\theta_{\eps J}^+) \ \cU_{J^c, l - \vert J \vert}
    \qquad
    (l \in \bN_n),
\label{H_l_gamma_0}
\ee
where
\be
    \cU_{J^c, l - \vert J \vert}
    = (-1)^{l - \vert J \vert}
        \sum_{\substack{I \subseteq J^c, \ \vert I \vert = l - \vert J \vert 
                        \\ \eps_j = \pm 1, \ j \in I}}
        \prod_{j \in I} e^{\eps_j \ri \nu}
        \prod_{\substack{j \in I, \ k \in J^c \setminus I \\ (j < k)}} 
            e^{\eps_j 2 \ri \mu}.
\label{cU}
\ee
By inspecting the above expression, let us observe that the value of 
$\cU_{J^c, l - \vert J \vert}$ does \emph{not} depend on the specific 
choice of the subset $J$, but only on its cardinality $\vert J \vert$. 
More precisely, if $J \subseteq \bN_n$ is an arbitrary subset of cardinality 
$\vert J \vert = k$ $(0 \leq k \leq l - 1)$, then we can write that
\be
    \cU_{J^c, l - \vert J \vert}
    = (-1)^{l - k}
        \sum_{\substack{1 \leq j_1 < \dots < j_{l - k} \leq n - k 
                        \\ \eps_1 = \pm 1, \dots, \eps_{l - k} = \pm 1}}
            \exp 
                \left(
                    \ri \sum_{m = 1}^{l - k} 
                        \eps_m \left( \nu + 2(n - l + m - j_m) \mu \right)
                \right).
\label{cU_ell-k}
\ee

To proceed further, let us now turn to the study of the Lax matrix $L$ 
\eqref{L}. Due to the Lax representation of the dynamics that we established 
in Theorem \ref{THEOREM_Lax_representation}, the eigenvalues of $L$ are 
conserved quantities. Consequently, the coefficients 
$K_0, K_1, \ldots, K_N \in C^\infty(P)$ of the characteristic polynomial 
\be
    \det(L - y \bsone_N) = \sum_{m = 0}^N K_{N - m} y^m
    \qquad
    (y \in \bC)
\label{L_kar_poly}
\ee
are also first integrals. As expected, the special algebraic properties 
of $L$ formulated in Proposition \ref{PROPOSITION_L_in_G} and Lemma 
\ref{LEMMA_L_in_exp_p} have a profound impact on these coefficients 
as well, as can be seen from the relations 
\be
    K_{N - m} = K_m
    \qquad 
    (m = 0, 1, \dots, N).
\label{K_m-symmetry}
\ee
So, it is enough to analyze the properties of the members 
$K_0 = 1, K_1, \ldots, K_n$. In this respect the most important ingredient 
is the relationship
\be
    \lim_{t \to \infty} L(\gamma(t)) = \exp(\bsTheta^+),
\label{asymptotic-Lax}
\ee
where $\bsTheta^+$ is the $N \times N$ diagonal matrix \eqref{bsTheta_+} 
containing the asymptotic momenta. Therefore, looking back to the definition
\eqref{L_kar_poly}, for any $m = 0, 1, \dots, n$ we obtain at once that 
\be
    K_m(\gamma_0)
    = \lim_{t \to \infty} K_m(\gamma(t))
    = (-1)^m
        \sum_{a = 0}^{\genfrac{\lfloor}{\rfloor}{}{}{m}{2}}
        \sum_{\substack{J \subseteq \bN_n, \ \vert J \vert = m - 2 a 
                        \\ \eps_j = \pm 1, \ j \in J}}
        \binom{n - \vert J \vert}{a} \cosh(\theta_{\eps J}^+).
\label{K_m_gamma_0}
\ee
Based on the formulae \eqref{H_l_gamma_0} and \eqref{K_m_gamma_0}, we can
prove the following important technical result.

\begin{LEMMA}
\label{LEMMA_linear_relation}
The two distinguished families of first integrals $\{ H_l \}_{l = 0}^n$ 
and $\{ K_m \}_{m = 0}^n$ are connected by an invertible linear relation with 
purely numerical coefficients depending only on the coupling parameters 
$\mu$ and $\nu$.
\end{LEMMA}

\begin{proof}
For brevity, let us introduce the notation
\be
    \cA_k 
    = \sum_{\substack{J \subseteq \bN_n, \ \vert J \vert = k 
                        \\ \eps_j = \pm 1, \ j \in J}}
            \cosh(\theta_{\eps J}^+)
    \qquad 
    (k = 0, 1, \dots, n).
\label{cA}
\ee
As we have seen in \eqref{cU_ell-k}, the coefficients 
$\cU_{J^c, l - \vert J \vert}$ appearing in the formula \eqref{H_l_gamma_0} 
depend only on the cardinality of $J$, whence for any 
$l \in \{ 0, 1, \dots, n \}$ we can write that
\be
    H_l(\gamma_0) = \sum_{k = 0}^l \cU_{\bN_{n - k}, l - k} \cA_k.
\label{H_vs_cA}
\ee
Since $\cU_{\bN_{n - l}, 0} = 1$, the matrix transforming
$\{ \cA_k \}_{k = 0}^n$ into $\{ H_l(\gamma_0) \}_{l = 0}^n$ is lower 
triangular with plus ones on the diagonal, whence the above linear 
relation \eqref{H_vs_cA} is invertible. Comparing the formulae 
\eqref{K_m_gamma_0} and \eqref{cA}, it is also clear that 
\be
    K_m(\gamma_0)
    = (-1)^m \sum_{a = 0}^{\genfrac{\lfloor}{\rfloor}{}{}{m}{2}}
                \binom{n - (m - 2 a)}{a} \cA_{m - 2 a},
\label{K_vs_cA}
\ee
which in turn implies that the matrix relating $\{ \cA_k \}_{k = 0}^n$ to 
$\{ K_m(\gamma_0) \}_{m = 0}^n$ is lower triangular with diagonal entries 
$\pm 1$. Hence the linear relationship \eqref{K_vs_cA} is also invertible. 
Putting together the above observations, it is clear that there is an 
invertible $(n + 1) \times (n + 1)$ matrix $\cC$ with purely numerical 
entries $\cC_{m, l}$ depending only on $\mu$ and $\nu$ such that
\be
    K_m(\gamma_0) = \sum_{l = 0}^n \cC_{m, l} H_l(\gamma_0).
\label{K_vs_H}
\ee
Since $\gamma_0$ is an arbitrary point of the phase space $P$ \eqref{P}, 
the Lemma follows.
\end{proof}

The scattering theoretical idea in the proof the above Lemma goes back 
to the fundamental works of Moser (see e.g. \cite{Moser_1975}). However, 
in the recent paper \cite{Gorbe_Feher_PLA2015} it has been revitalized in 
the context of the rational $BC_n$ van Diejen model, too. Compared to the 
rational case, it is a significant difference that our coefficients 
$\cU_{J^c, l - \vert J \vert}$ \eqref{cU} do depend on the parameters 
$\mu$ and  $\nu$ in a non-trivial manner, whence the observations surrounding 
the derivations of formula \eqref{cU_ell-k} turns out to be crucial in our 
presentation.

Since the family of functions $\{ H_l \}_{l = 0}^n$ Poisson commute, Lemma 
\ref{LEMMA_linear_relation} readily implies that the first integrals 
$\{ K_m \}_{m = 0}^n$ are also in involution. Now, let us recall that the 
spectrum of the Lax matrix $L$ is simple, as we have seen in Lemma 
\ref{LEMMA_regularity}. As a consequence, the eigenvalues of $L$ can be 
realized as smooth functions of the coefficients of the characteristic 
polynomial \eqref{L_kar_poly}, thus the following result is immediate.

\begin{THEOREM}
\label{THEOREM_commuting_eigenvalues}
The eigenvalues of the Lax matrix $L$ \eqref{L} are in involution.
\end{THEOREM}

To conclude this section, let us note that the proof of Theorem 
\ref{THEOREM_commuting_eigenvalues} is quite indirect in the sense that
it hinges on the commutativity of the family of functions \eqref{H_vD}.
However, the only available proof of this highly non-trivial fact is
based on the observation that the Hamiltonians \eqref{H_vD} can be
realized as classical limits of van Diejen's commuting analytic difference 
operators \cite{van_Diejen_ComposMath}. As a more elementary approach, let 
us note that Theorem \ref{THEOREM_commuting_eigenvalues} would also follow 
from the existence of an $r$-matrix encoding the tensorial Poisson bracket 
of the Lax matrix $L$ \eqref{L}. Due to Lemma \ref{LEMMA_linear_relation}, 
it would imply the commutativity of the family \eqref{H_vD}, too, at least 
under the specialization \eqref{2parameters}. To find such an $r$-matrix, 
one may wish to generalize the analogous results on the rational system 
\cite{Pusztai_NPB2015}.

\section{Discussion}
\label{SECTION_Discussion}
One of the most important objects in the study of integrable systems is the
Lax representation of the dynamics. By generalizing the earlier results on 
the rational $BC_n$ RSvD models \cite{Pusztai_NPB2011, Pusztai_NPB2015},
in this paper we succeeded in constructing a Lax pair for the $2$-parameter 
family of hyperbolic van Diejen systems \eqref{H}. Making use of this 
construction, we showed that the dynamics can be solved by a projection 
method, which in turn allowed us to initiate the study of the scattering 
properties of \eqref{H}. Moreover, by combining our scattering theoretical 
results with the ideas of the recent paper \cite{Gorbe_Feher_PLA2015}, we 
proved that the first integrals provided by the eigenvalues of the proposed 
Lax matrix \eqref{L} are in fact in involution. To sum up, it is fully 
justified to say that the matrices $L$ \eqref{L} and $B$ \eqref{B} form a 
Lax pair for the hyperbolic van Diejen system \eqref{H}. 

Apart from taking a non-trivial step toward the construction of Lax matrices 
for the most general hyperbolic van Diejen many-particle systems \eqref{H_vD}, 
let us not forget about the potential applications of our results. In analogy 
with the translation invariant RS systems, we expect that the van Diejen 
models may play a crucial role in clarifying the particle-soliton picture 
in the context of integrable boundary field theories. While the relationship 
between the $A$-type RS models and the soliton equations defined on the whole 
line is under control (see e.g. \cite{Ruij_Schneider, Ruij_CMP1988, 
Babelon_Bernard, Ruij_RIMS_2, Ruij_RIMS_3}), the link between the van Diejen 
models and the soliton systems defined on the half-line is less understood 
(see e.g. \cite{Saleur_et_al_NPB1995, Kapustin_Skorik}). As in the translation 
invariant case, the Lax matrices of the van Diejen systems could turn out to 
be instrumental for elaborating this correspondence.

Turning to the more recent activities surrounding the CMS and the RS 
many-particle models, let us recall the so-called classical/quantum duality 
(see e.g. \cite{Mukhin_et_al_2011, Alexandrov_et_al_NPB2014, 
Gorsky_et_al_JHEP2014, Tsuboi_et_al_JHEP2015, Beketov_et_al_NPB2016}), which 
relates the spectra of certain quantum spin chains with the Lax matrices of 
the classical CMS and RS systems. An equally remarkable development is the 
emergence of new integrable tops based on the Lax matrices of the CMS and 
the RS systems \cite{Aminov_et_al_2014, Levin_et_al_JHEP2014}. Relatedly, 
it would be interesting to see whether the Lax matrix \eqref{L} of the 
hyperbolic van Diejen system \eqref{H} can be fit into these frameworks.

One of the most interesting aspects of the CMS and the RSvD systems we 
have not addressed in this paper is the so-called Ruijsenaars duality, or 
action-angle duality. Based on hard analytical techniques, this remarkable 
property was first exhibited by Ruijsenaars \cite{Ruij_CMP1988} in the 
context of the translation invariant non-elliptic models. Let us note 
that in the recent papers \cite{Feher_Klimcik_0901, Feher_Ayadi_JMP2010,
Feher_Klimcik_CMP2011, Feher_Klimcik_NPB2012} almost all of these duality 
relationships have been successfully reinterpreted in a nice geometrical 
framework provided by powerful symplectic reduction methods. Moreover, 
by now some duality results are available also for the CMS and the RSvD 
models associated with the $BC$-type root systems \cite{Pusztai_NPB2011, 
Pusztai_NPB2012, Feher_Gorbe_JMP2014}. 

As for the key player of our paper, we have no doubt that the $2$-parameter
family of hyperbolic van Diejen systems \eqref{H} is self-dual. Indeed, upon 
diagonalizing the Lax matrix $L$ \eqref{L}, we see that the transformed 
objects defined in \eqref{L_diagonalized}-\eqref{hat_L_and_hat_F} obey the 
relationship \eqref{hat_commut_rel}, that has the same form as the Ruijsenaars 
type commutation relation \eqref{commut_rel} we set up in Lemma 
\ref{LEMMA_commut_rel}. Based on the method presented in \cite{Ruij_CMP1988}, 
we expect that the transformed matrix $\hat{L}$ \eqref{hat_L_and_hat_F} shall 
provide a Lax matrix for the dual system. Therefore, comparing the matrix 
entries displayed in \eqref{L} and \eqref{hat_L_entries}, the self-duality 
of the system \eqref{H} seems to be more than plausible. Admittedly, many 
subtle details are still missing for a complete proof. As for filling these 
gaps, the immediate idea is that either one could mimic Ruijsenaars' 
scattering theoretical approach, or invent an appropriate symplectic reduction 
framework. However, notice that the non-standard form of the Hamiltonian 
\eqref{H} poses severe analytical difficulties on the study of the scattering 
theory, whereas the weakness of the geometrical approach lies in the fact that 
up to now even the translation invariant hyperbolic RS model has not been 
derived from symplectic reduction. Nevertheless, by taking the analytical 
continuation of the Lax matrix $L$ \eqref{L}, it is conceivable that the 
self-duality of the compactified trigonometric version of \eqref{H} can be 
proved by adapting the quasi-Hamiltonian reduction approach advocated by 
Feh\'er and Klim\v{c}\'ik \cite{Feher_Klimcik_NPB2012}. For further 
motivation, let us recall that the duality properties are indispensable in 
the study of the recently introduced integrable random matrix ensembles 
\cite{Bogomolny_et_al_PRL2009, Bogomolny_et_al_Nonlinearity2011, 
Fyodorov_Giraud_2015}, too.

From the above paragraphs it is clear that our results on the $2$-parameter 
family of hyperbolic systems \eqref{H} open up a plethora of interesting 
problems. Besides, based on our numerical calculations, below we also wish to 
discuss some possible generalizations in two further directions. First, it is 
a time-honored principle that the inclusion of a spectral parameter into the 
Lax matrix of an integrable system can greatly enrich the analysis by 
borrowing techniques from complex geometry. Bearing this fact in mind, with 
the aid of the function
\be
    \Phi(x \, | \, \eta) 
    = e^{x \coth(\eta)} \left( \coth(x) - \coth(\eta) \right)
\label{Phi}
\ee
depending on the complex variables $x$ and $\eta$, over the phase space 
$P$ \eqref{P} we define the matrix valued smooth function 
$\cL = \cL(\lambda, \theta; \mu, \nu \, | \, \eta)$ with entries
\be
    \cL_{k, l}
    = \left(
        \ri \sin(\mu) F_k \bar{F}_l + \ri \sin(\mu - \nu) C_{k, l})
    \right)
    \Phi(\ri \mu + \Lambda_j - \Lambda_k \, | \, \eta)
    \qquad
    (k, l \in \bN_N).
\label{cL_matrix}
\ee
One of the outcomes of our numerical investigations is that for any values 
of $\eta$ the eigenvalues of $\cL$ provide a family of first integrals in 
involution for the van Diejen system \eqref{H}. Thinking of $\eta$ as a 
spectral parameter, let us also observe that, in the limit 
$\bR \ni \eta \to \infty$, from $\cL$ we can recover our Lax matrix $L$ 
\eqref{L}; that is, $\cL \to L$. Although the spectral parameter dependent 
matrix $\cL$ does not take values in the Lie group $U(n, n)$ \eqref{G}, we 
find it interesting that the constituent function $\Phi$ \eqref{Phi} can 
be seen as a hyperbolic limit of the elliptic Lam\'e function, that plays 
a prominent role in the theory of the elliptic CMS and RS systems (see e.g. 
the papers \cite{Krichever, Ruij_CMP1987} and the monograph 
\cite{Babelon_Bernard_Talon}). Therefore, it is tempting to think that an 
appropriate elliptic deformation of $\cL$ \eqref{L} may lead to a spectral 
parameter dependent Lax matrix of the elliptic van Diejen system with 
coupling parameters $\mu$ and $\nu$.

Hitherto we have studied the van Diejen system \eqref{H} with only two 
independent coupling parameters. Though a construction of a Lax matrix 
for the most general hyperbolic van Diejen system with five independent 
coupling parameters still seems to be out of reach, we can offer a plausible 
conjecture for a Lax matrix with \emph{three} independent coupling constants. 
Simply by generalizing the formulae appearing in the theory of the rational 
$BC_n$ RSvD systems \cite{Pusztai_NPB2013}, with the aid of an additional 
real parameter $\kappa$ let us define the real valued functions $\alpha$ 
and $\beta$ for any $x > 0$ by the formulae
\be
    \alpha(x)
    = \frac{1}{\sqrt{2}} 
        \left(
            1 + \left( 1 + \frac{\sin(\kappa)^2}{\sinh(2 x)^2} \right)^\half 
        \right)^\half
    \quad \text{and} \quad
    \beta(x) 
    = \frac{\ri}{\sqrt{2}}
        \left(
            -1 + \left( 1 + \frac{\sin(\kappa)^2}{\sinh(2 x)^2} \right)^\half
        \right)^\half.
\label{alpha_beta}
\ee
Built upon these functions, let us also introduce the Hermitian $N \times N$ 
matrix
\be
    h(\lambda) 
        = \begin{bmatrix}
        \diag(\alpha(\lambda_1), \ldots, \alpha(\lambda_n)) 
            & \diag(\beta(\lambda_1), \ldots, \beta(\lambda_n)) 
        \\
        -\diag(\beta(\lambda_1), \ldots, \beta(\lambda_n)) 
            & \diag(\alpha(\lambda_1), \ldots, \alpha(\lambda_n))
        \end{bmatrix}.
\label{h}
\ee
One can easily show that $h C h = C$, whence the matrix valued function
\be
    \tilde{L} = h^{-1} L h^{-1}
\label{L_tilde}
\ee
also takes values in the Lie group $U(n, n)$ \eqref{G}. Notice that the 
rational limit of matrix $\tilde{L}$ gives back the Lax matrix of the 
rational $BC_n$ RSvD system, that first appeared in equation (4.51) of 
paper \cite{Pusztai_NPB2012}. Moreover, upon setting
\be
    g = \mu,
    \quad
    g_0 = g_1 = \frac{\nu}{2},
    \quad
    g'_0 = g'_1 = \frac{\kappa}{2},
\label{3parameters}
\ee
for van Diejen's main Hamiltonian $H_1$ \eqref{H_vD} we get that
\be
    H_1 
    = 2 \sum_{a=1}^n
        \cosh(\theta_a) u_a 
        \left(
            1 + \frac{\sin(\kappa)^2}{\sinh(2 \lambda_a)^2}
        \right)^\half
    + 2 \sum_{a=1}^n
        \Real 
        \left( 
            z_a \frac{\sinh(\ri \kappa + 2 \lambda_a)}{\sinh(2 \lambda_a)}
        \right),
\label{H_1_mu_nu_kappa}
\ee
with the functions $z_a$ and $u_a$ defined in the equations \eqref{z_a} and 
\eqref{u_a}, respectively. The point is that, in complete analogy with 
\eqref{H1_vs_H}, one can establish the relationship
\be
    H_1 
    + 2 \cos 
        \left( 
            \nu + \kappa + (n - 1) \mu 
        \right) 
        \frac{\sin(n \mu)}{\sin(\mu)}
    = \tr(\tilde{L}).
\label{H1_vs_H_kappa}
\ee
Furthermore, based on numerical calculations for small values of $n$, 
it appears that the eigenvalues of $\tilde{L}$ \eqref{L_tilde} provide 
a commuting family of first integrals for the van Diejen system 
\eqref{H_1_mu_nu_kappa}. To sum up, we have numerous evidences that matrix 
$\tilde{L}$ \eqref{L_tilde} is a Lax matrix for the $3$-parameter family 
of van Diejen systems \eqref{H_1_mu_nu_kappa}, if the pertinent parameters 
are connected by the relationships displayed in \eqref{3parameters}.
As can be seen in \cite{Pusztai_NPB2012}, the new parameter $\kappa$ causes
many non-trivial technical difficulties even at the level of the rational
van Diejen system. Part of the difficulties can be traced back to the fact 
that for $\sin(\kappa) \neq 0$ the matrix $\tilde{L}$ \eqref{L_tilde} does 
not belong to the symmetric space $\exp(\mfp)$ \eqref{symm_space}, whence 
the diagonalization of $\tilde{L}$ requires a less direct approach than 
that provided by the canonical form \eqref{mfp_and_mfa_and_K}. We wish to 
come back to these problems in later publications.

\medskip
\noindent
\textbf{Acknowledgments.}
We are grateful to L.~Feh\'er (Univ. Szeged) for useful comments on the 
manuscript. The work of B.G.P. was supported by the J\'anos Bolyai Research 
Scholarship of the Hungarian Academy of Sciences, by the Hungarian 
Scientific Research Fund (OTKA grant K116505), and also by a Lend\"ulet 
Grant; he wishes to thank to Z.~Bajnok for hospitality in the MTA Lend\"ulet 
Holographic QFT Group. T.F.G. was supported in part by the Hungarian 
Scientific Research Fund (OTKA grant K111697) and by COST (European 
Cooperation in Science and Technology) in COST Action MP1405 QSPACE.



\begin{thebibliography}{99}

    \bibitem{van_Diejen_ComposMath}
        J.F.~van Diejen,
        Commuting difference operators with polynomial eigenfunctions,
        \emph{Compositio Math.} \textbf{95} (1995) 183-233.

    \bibitem{van_Diejen_TMP1994}
        J.F.~van Diejen,
        Deformations of Calogero--Moser systems,
        \emph{Theor. Math. Phys.} {\bf 99} (1994) 549-554.

    \bibitem{van_Diejen_JMP1995}
        J.F.~van Diejen,
        Difference Calogero--Moser systems and finite Toda chains,
        \emph{J. Math. Phys.} {\bf 36} (1995) 1299-1323.

	\bibitem{Ruij_Schneider}
		S.N.M.~Ruijsenaars and H.~Schneider,
		A new class of integrable models and its relation to solitons,
		\emph{Ann. Phys. (N.Y.)} \textbf{170} (1986) 370-405.

    \bibitem{Ruij_CMP1987}
        S.N.M.~Ruijsenaars, 
        Complete integrability of relativistic Calogero--Moser systems 
        and elliptic function identities, 
        \emph{Commun. Math. Phys.} \textbf{110} (1987) 191-213.

    \bibitem{Calogero}
		F.~Calogero, Solution of the one-dimensional $N$-body problem with 
		quadratic and/or inversely quadratic pair potentials, 
		\emph{J. Math. Phys.} \textbf{12} (1971) 419-436.
	
	\bibitem{Sutherland}
		B.~Sutherland, 
		Exact results for a quantum many body problem in one dimension, 
		\emph{Phys. Rev.} \textbf{A 4} (1971) 2019-2021.
	
	\bibitem{Moser_1975}
		J.~Moser,
		Three integrable Hamiltonian systems connected with isospectral 
		deformations, 
		\emph{Adv. Math.} \textbf{16} (1975) 197-220.
		
	\bibitem{Olsha_Pere_1976}
        M.A.~Olshanetsky and A.M.~Perelomov,
        Completely integrable Hamiltonian systems connected with semisimple 
        Lie algebras,
        \emph{Invent. Math.} \textbf{37} (1976) 93-108.

    \bibitem{Pusztai_NPB2011}
        B.G.~Pusztai,
        Action-angle duality between the $C_n$-type hyperbolic Sutherland 
        and the rational Ruijsenaars--Schneider--van Diejen models,
        \emph{Nucl. Phys.} \textbf{B 853} (2011) 139-173.

    \bibitem{Pusztai_NPB2012}
        B.G.~Pusztai,
        The hyperbolic $BC_n$ Sutherland and the rational $BC_n$ 
        Ruijsenaars--Schnei\-der--van Diejen models: Lax matrices and duality,
        \emph{Nucl. Phys.} \textbf{B 856} (2012) 528-551.

    \bibitem{Pusztai_NPB2013}
        B.G.~Pusztai,
        Scattering theory of the hyperbolic $BC_n$ Sutherland and
        the rational $BC_n$ Ruijsenaars--Schneider--van Diejen models,
        \emph{Nucl. Phys.} \textbf{B 874} (2013) 647-662.

    \bibitem{Feher_Gorbe_JMP2014}
        L.~Feh\'er and T.F.~G\"orbe,
        Duality between the trigonometric $BC_n$ Sutherland system and 
        a completed rational Ruijsenaars--Schneider--van Diejen system,
        \emph{J. Math. Phys.} \textbf{55} (2014) 102704.

    \bibitem{Gorbe_Feher_PLA2015}
        T.F.~G\"orbe and L.~Feh\'er,
        Equivalence of two sets of Hamiltonians associated with the rational
        $BC_n$ Ruijsenaars--Schneider--van Diejen system,
        \emph{Phys. Lett.} \textbf{A 379} (2015) 2685-2689.

    \bibitem{Pusztai_NPB2015}
        B.G.~Pusztai,
        On the classical $r$-matrix structure of the rational $BC_n$ 
        Ruijsenaars--Schnei\-der--van Diejen system,
        \emph{Nucl. Phys.} \textbf{B 900} (2015) 115-146.

    \bibitem{Knapp}
        A.W.~Knapp,
        \emph{Lie groups beyond an introduction},
        Progress in Mathematics, vol. 140, Birkh\"auser, Boston, MA, 2002.

    \bibitem{Ruij_CMP1988}
        S.N.M.~Ruijsenaars, 
        Action-angle maps and scattering theory for some finite dimensional 
        integrable systems I. The pure soliton case,
        \emph{Commun. Math. Phys.} \textbf{115} (1988) 127-165.

    \bibitem{AM}
        R.~Abraham and J.E.~Marsden,
        \emph{Foundations of Mechanics}, second ed.,
        Addison Wesley, 1985.

    \bibitem{OlshaPere}
        M.A.~Olshanetsky and A.M.~Perelomov,
        Classical integrable finite-dimensional systems related to 
        Lie algebras, 
        \emph{Phys. Rep.} \textbf{71} (1981) 313-400.

    \bibitem{KKS}
        D.~Kazhdan, B.~Kostant and S.~Sternberg, 
        Hamiltonian group actions and dynamical systems of Calogero type, 
        \emph{Commun. Pure Appl. Math.} \textbf{XXXI} (1978) 481-507.
        
    \bibitem{Feher_Pusztai_NPB2006}
        L.~Feh\'er and B.G.~Pusztai,
        Spin Calogero models associated with Riemannian symmetric spaces of 
        negative curvature,
        \emph{Nucl. Phys.} \textbf{B 751} (2006) 436-458.

    \bibitem{Feher_Pusztai_2007}
        L.~Feh\'er and B.G.~Pusztai,
        A class of Calogero type reductions of free motion 
        on a simple Lie group,
        \emph{Lett. Math. Phys.} \textbf{79} (2007) 263-277.
        
    \bibitem{Feher_Klimcik_0901}
        L.~Feh\'er and C.~Klim\v{c}\'ik,
        On the duality between the hyperbolic Sutherland and the rational 
        Ruijsenaars--Schneider models,
        \emph{J. Phys. A: Math. Theor.} \textbf{42} (2009) 185202.

    \bibitem{Ruij_FiniteDimSolitonSystems}
        S.N.M.~Ruijsenaars,
        Finite-dimensional soliton systems, 
        in: B. Kupershmidt (Ed.), \emph{Integrable and superintegrable 
        systems}, 
        World Scientific, 1990, pp. 165-206.
        
	\bibitem{Kulish_1976}
		P.P.~Kulish,
		Factorization of the classical and the quantum $S$ matrix and 
		conservation laws,
		\emph{Theor. Math. Phys.} \textbf{26} (1976) 132-137.
		
	\bibitem{Moser_1977}
		J.~Moser,
		The scattering problem for some particle systems on the line, 
		in: \emph{Lecture Notes in Mathematics} \textbf{597}, 
		Springer, 1977, pp. 441-463.

    \bibitem{Babelon_Bernard}
		O.~Babelon and D.~Bernard,
		The sine-Gordon solitons as an $N$-body problem,
		\emph{Phys. Lett.} \textbf{B 317} (1993) 363-368.
        
    \bibitem{Ruij_RIMS_2}
        S.N.M.~Ruijsenaars, 
        Action-angle maps and scattering theory for some finite dimensional 
        integrable systems II. Solitons, antisolitons and their bound states,
        \emph{Publ. RIMS} \textbf{30} (1994) 865-1008.

    \bibitem{Ruij_RIMS_3}
        S.N.M.~Ruijsenaars, 
        Action-angle maps and scattering theory for some finite dimensional 
        integrable systems III. Sutherland type systems and their duals,
        \emph{Publ. RIMS} \textbf{31} (1995) 247-353.

    \bibitem{Saleur_et_al_NPB1995}
        H.~Saleur, S.~Skorik and N.P.~Warner,
        The boundary sine-Gordon theory: classical and semi-classical 
        analysis,
        \emph{Nucl. Phys.} \textbf{B 441} (1995) 421-436.

    \bibitem{Kapustin_Skorik}
        A.~Kapustin and S.~Skorik,
        On the non-relativistic limit of the quantum sine-Gordon model with 
        integrable boundary condition,
        \emph{Phys. Lett.} \textbf{A 196} (1994) 47-51.
        
    \bibitem{Mukhin_et_al_2011}
        E.~Mukhin, V.~Tarasov and A.~Varchenko,  
        Gaudin Hamiltonians generate the Bethe algebra of a tensor power 
        of the vector representation of $\mfgl_N$,
        \emph{St. Petersburg Math. J.} \textbf{22} (2011) 463-472.

    \bibitem{Alexandrov_et_al_NPB2014}
        A.~Alexandrov, S.~Leurent, Z.~Tsuboi and A.~Zabrodin,
        The master $T$-operator for the Gaudin model and the KP hierarchy,
        \emph{Nucl. Phys.} \textbf{B 883} (2014) 173-223.
	    
    \bibitem{Gorsky_et_al_JHEP2014}
        A.~Gorsky, A.~Zabrodin and A.~Zotov,
        Spectrum of Quantum Transfer Matrices via Classical Many-Body Systems,
        \emph{JHEP} \textbf{01} (2014) 070. 

    \bibitem{Tsuboi_et_al_JHEP2015}
        Z.~Tsuboi, A.~Zabrodin and A.~Zotov,
        Supersymmetric quantum spin chains and classical integrable systems,
        \emph{JHEP} \textbf{05} (2015) 086.
        
    \bibitem{Beketov_et_al_NPB2016}
        M.~Beketov, A.~Liashyk, A.~Zabrodin and A.~Zotov,
        Trigonometric version of quantum-classical duality in 
        integrable systems,
        \emph{Nucl. Phys.} \textbf{B 903} (2016) 150-163.
        
    \bibitem{Aminov_et_al_2014}
        G.~Aminov, S.~Arthamonov, A.~Smirnov and A.~Zotov,
        Rational Top and its Classical $R$-mat\-rix,
        \emph{J. Phys. A: Math. Theor.} \textbf{47} (2014) 305207.
        
    \bibitem{Levin_et_al_JHEP2014}
        A.~Levin, M.~Olshanetsky and A.~Zotov, 
        Relativistic Classical Integrable Tops and Quantum $R$-matrices,
        \emph{JHEP} \textbf{07} (2014) 012.
        
    \bibitem{Feher_Ayadi_JMP2010}
        L.~Feh\'er and V.~Ayadi,
        Trigonometric Sutherland systems and their Ruijsenaars duals from
        symplectic reduction,
        \emph{J. Math. Phys.} \textbf{51} (2010) 103511.
    
    \bibitem{Feher_Klimcik_CMP2011}
        L.~Feh\'er and C.~Klim\v{c}\'ik,
        Poisson--Lie interpretation of trigonometric Ruijsenaars duality,
        \emph{Commun. Math. Phys.} \textbf{301} (2011) 55-104. 
        
    \bibitem{Feher_Klimcik_NPB2012}
        L.~Feh\'er and C.~Klim\v{c}\'ik,
        Self-duality of the compactified Ruijsenaars--Schneider system 
        from quasi-Hamiltonian reductions,
        \emph{Nucl. Phys.} \textbf{B 860} (2012) 464-515.

    \bibitem{Bogomolny_et_al_PRL2009}
        E.~Bogomolny, O.~Giraud and C.~Schmit,
        Random Matrix Ensembles Associated with Lax Matrices,
        \emph{Phys. Rev. Lett.} \textbf{103} (2009) 054103.
        
    \bibitem{Bogomolny_et_al_Nonlinearity2011}
        E.~Bogomolny, O.~Giraud and C.~Schmit,
        Integrable random matrix ensembles,
        \emph{Nonlinearity} \textbf{24} (2011) 3179-3213.
		
	\bibitem{Fyodorov_Giraud_2015}
	    Y.V.~Fyodorov and O.~Giraud,
	    High values of disorder-generated multifractals and logarithmically 
	    correlated processes,
	    \emph{Chaos, Solitons \& Fractals} \textbf{74} (2015) 15-26.

    \bibitem{Krichever}
	    I.M.~Krichever,
		Elliptic solutions of the Kadomtsev--Petviashvili equation
        and integrable systems of particles, 
        \emph{Funct. Anal. Appl.} \textbf{14} (1980) 282-290.

    \bibitem{Babelon_Bernard_Talon}
        O.~Babelon, D.~Bernard and M.~Talon,
        \emph{Introduction to classical integrable systems},
        Cambridge University Press, 2003.

\end{thebibliography}
\end{document}